\documentclass{article}
\usepackage{ijcai25}

\usepackage{times}
\usepackage{soul}
\usepackage{url}
\usepackage[hidelinks]{hyperref}
\usepackage[utf8]{inputenc}
\usepackage[small]{caption}
\usepackage{graphicx}
\usepackage{amsmath}
\usepackage{amsthm}
\usepackage{algorithm}
\usepackage{algorithmic}
\usepackage[switch]{lineno}
\usepackage{xcolor}
\usepackage{amssymb}
\usepackage{tikz}
\usepackage{booktabs}
\usepackage{makecell}
\usepackage{diagbox}
\newtheorem{definition}{Definition}
\usepackage{subfigure}

\usepackage{enumerate}
\usepackage{verbatim}
\newtheorem{theorem}{Theorem}
\newtheorem{lemma}{Lemma}
\newtheorem{corollary}{Corollary}

\title{Deterministic implementation in single-item auctions}

\author{
Yan Liu$^1$ \and
Zeyu Ren$^1$ \and
Pingzhong Tang$^{2}$ \and
Zihe Wang$^1$ \and
Yulong Zeng$^3$ \And
Jie Zhang$^4$
\affiliations
$^1$Renmin University of China \\
$^2$Tsinghua University \\
$^3$Tsinghua University \\
$^4$University of bath \\
\emails
liuyan5816@ruc.edu.cn,
zeyuren@ruc.edu.cn,
kenshinping@gmail.com,
wang.zihe@ruc.edu.cn,
freeof123@qq.com,
jz2558@bath.ac.uk
}

\begin{document}

\maketitle

\begin{abstract}
    Deterministic auctions are attractive in practice due to their transparency, simplicity, and ease of implementation, motivating a sharper understanding of when they can attain the same outcomes as randomized mechanisms.
    We study deterministic implementation in single-item auctions under two notions of outcomes: (revenue, welfare) pairs and interim allocations.
    For (revenue, welfare) pairs, we show a separation in discrete settings: there exists a pair implementable by a deterministic Bayesian incentive-compatible (BIC) auction but not by any deterministic dominant-strategy incentive-compatible (DSIC) auction.
    For continuous atomless priors, we identify conditions under which deterministic DSIC auctions are equivalent to randomized BIC auctions in terms of achievable outcomes.
    For interim allocations, under a strict monotonicity condition, we establish a deterministic analogue of Border's theorem for two bidders, providing a necessary and sufficient condition for deterministic DSIC implementability.
    Using this characterization, we exhibit an interim allocation implementable by a randomized BIC auction but not by any deterministic DSIC auction.
\end{abstract}

\section{Introduction}
The central objective of mechanism design is to identify and characterize the set of outcomes that can be implemented by suitable mechanisms.
Since the foundational contributions of \cite{hurwicz1960optimality,hurwicz1972informationally}, a substantial literature has developed across diverse market environments and under multiple solution concepts \cite{clarke1971multipart,dobzinski2022hardness,feldman2023simultaneous,groves1973incentives,myerson1979incentive,myerson1981optimal,myerson1982optimal,vickrey1961counterspeculation,you2021incentive}.

This paper advances that agenda in the setting of single-item auctions by studying the implementability of deterministic auctions, where allocation and payment rules are deterministic.
Determinism is important in many applications, including antique auctions, spectrum sales, and platform marketplaces, where transparency, compliance, or operational simplicity may rule out randomization \cite{klemperer2004auctions,liu2024optimal}.
Recent work on single-item auctions shows that, for revenue maximization, optimal auctions can often be chosen as deterministic DSIC mechanisms in discrete settings, despite allowing randomized BIC mechanisms \cite{giannakopoulos2024discrete}.
These optimality results, however, leave open a complementary question that is central to implementability: beyond the outcomes induced by optimizing a fixed objective, when does determinism preserve (or fail to preserve) the full set of achievable trade-offs and reduced forms?
A precise account of the power and limitations of deterministic auctions is therefore both practically relevant and theoretically significant.

We take two complementary views.
First, an outcome-space view that charts the set of achievable (revenue, welfare) pairs, together with its Pareto frontier.
Second, a reduced-form view that characterizes the interim allocations implementable by deterministic mechanisms.
These two views allow us to assess determinism at the level of aggregate trade-offs and at the level of allocation structure.

\subsection{Implementation with Respect to (Revenue, Welfare) Pairs}
Our first objective is to characterize which (revenue, welfare) pairs are attainable by deterministic DSIC auctions.
For randomized auctions, the picture is well understood: one first identifies the extreme points and then takes convex combinations to generate the full attainable set.
\cite{myerson1983efficient} characterize the extreme points for the single-buyer case.
\cite{likhodedov2003mechanism} extend this to multiple buyers.
This program does not carry over to deterministic auctions, because a convex combination of two deterministic auctions is not necessarily deterministic.
As a result, the results and techniques of \cite{likhodedov2003mechanism} do not apply, and the set of extreme points for deterministic auctions is not continuous.

We show that deterministic DSIC and deterministic BIC are not equivalent with respect to (revenue, welfare) pairs.
There exists a pair that is implementable by a deterministic BIC auction but not by any deterministic DSIC auction.
Prior results on revenue inequivalence between DSIC and BIC do not apply here, since they are confined to two-item environments \cite{tang2016optimal,yao2017dominant}.
To further delineate the boundary of this separation, we identify a restricted class of distributions under which randomized BIC auctions and deterministic DSIC auctions remain equivalent with respect to (revenue, welfare) pairs.
In the continuous case, we obtain two levels of equivalence: under regularity we match the Pareto frontier, and under stronger distributional conditions (in the i.i.d. case) we match the entire attainable region.

\subsection{Implementation with Respect to Interim Allocations}
Our second objective is to characterize which interim allocations can be implemented by deterministic auctions.
Here, ``interim allocation'' refers to the expected allocation probability, taken over the randomness in the other players' types, in the same sense as defined in \cite{myerson1981optimal}.
\cite{border1991implementation} gives a necessary and sufficient condition, now called Border's theorem, for when an interim allocation (the ``reduced form'' in Border's terminology) is implementable by a randomized BIC auction.
Building on this, \cite{manelli2010bayesian} show that any interim allocation of a randomized BIC auction is also implementable by a randomized DSIC auction.
Further, \cite{chen2019equivalence} prove that for continuous atomless priors, every interim allocation implementable by a randomized BIC auction is also implementable by a deterministic BIC auction.
To the best of our knowledge, no analogous result is known for deterministic DSIC auctions.
We show that an analogous equivalence fails for deterministic DSIC auctions.
To show this separation, we introduce new techniques, most notably a deterministic version of Border's theorem.

\subsection{Our Results}
We study the power and limitations of deterministic auctions in single-item environments through two outcome notions: (revenue, welfare) pairs and interim allocations.
\begin{itemize}
    \item Revenue-welfare outcomes.
    We characterize which (revenue, welfare) pairs are attainable by deterministic DSIC auctions.
    In discrete settings, we exhibit a counterexample (via computational enumeration) showing a pair that is implementable by a deterministic BIC auction but not by any deterministic DSIC auction, establishing a non-equivalence between the two.
    In contrast, under continuous regular priors, every Pareto-optimal (revenue, welfare) pair achievable by randomized BIC auctions is implementable by a deterministic DSIC auction.
    Moreover, in the i.i.d. continuous case, under Condition (\ref{thrid_order}), every (revenue, welfare) pair implementable by randomized BIC auctions is also implementable by a deterministic DSIC auction.
    \item Interim allocations.
    We provide a deterministic analogue of Border's theorem for two bidders under a strict monotonicity condition, yielding a necessary and sufficient condition under which an interim allocation is implementable by a deterministic DSIC auction.
    Using this characterization, we show that randomized BIC and deterministic DSIC auctions are not equivalent.
    For three bidders with constant interim allocations, we further derive a necessary and sufficient condition for implementability under deterministic DSIC.
\end{itemize}

\subsection{Related Work}
\cite{diakonikolas2012efficiency} define the subset of (revenue, welfare) pairs that are undominated in both coordinates as Pareto points, and call their collection the Pareto curve, i.e., the upper-right boundary of the attainable set.
They show that even with two buyers facing asymmetric discrete priors, this curve need not be concave, and that deciding whether a given point is achievable by a deterministic auction is NP-hard.
In contrast, we prove that under continuous and regular priors, the Pareto frontier of randomized BIC mechanisms coincides with that of deterministic DSIC mechanisms.

\cite{arigapudi2018equivalence} studies the related implementability problem in a symmetric environment.
They give necessary and sufficient conditions under which the interim allocation of a symmetric BIC auction is implementable by a symmetric DSIC auction.
We consider the two-buyer case without symmetry and our results specialize to theirs.
Moreover, our theorem (Theorem \ref{Thm:Border_deterministic}) does not assume that the target interim allocation is implementable by any deterministic BIC mechanism.
On the randomized side, \cite{border2007reduced} extend Border's theorem to asymmetric settings, and \cite{che2013generalized} provide a clean network-flow proof.
For deterministic implementability, \cite{bikhchandani2006weak} show that a deterministic social choice function over multi-dimensional types is DSIC if and only if it is weakly monotone.
\cite{gershkov2013equivalence} further broaden equivalence results for randomized BIC and randomized DSIC mechanisms.

The study of attainable (revenue, welfare) outcomes is a recurring theme in auction theory (e.g., \cite{bergemann2015limits,ma2022revenue}).
\cite{hartline2008optimal} design an auction that maximizes buyer surplus (welfare minus revenue), which corresponds to an extreme point of the outcome space.
\cite{pai2009competing} analyzes competition among sellers, yielding a revenue-utility trade-off for buyers.
\cite{kleinberg2013ratio} establish lower bounds on the revenue-to-welfare ratio under various feasibility constraints.
In a similar spirit, \cite{bachrach2014optimising} optimize a composite objective ``$\lambda_1 \cdot \mathrm{revenue}$ + $\lambda_2 \cdot \mathrm{welfare}$ + $\lambda_3 \cdot \mathrm{click ~ yield}$'' to balance stakeholder goals, and \cite{shen2017practical} propose a parametric family that optimizes a linear combination of revenue and welfare inside monotone transformations, achieving near-optimal revenue with practical implementability.

\cite{giannakopoulos2024discrete} study discrete single-item optimal auction design and show that, for revenue maximization (and more generally linear objectives), an optimal mechanism can be chosen deterministic DSIC via an LP-duality/KKT characterization.
Our paper is complementary: instead of identifying an optimizer for a fixed objective, we characterize which (revenue, welfare) trade-offs and interim allocations are implementable under determinism, and we obtain separations between deterministic BIC and deterministic DSIC.
Since their techniques are tailored to optimality for a given objective, they do not imply our implementability questions or preclude the separations we establish.

\section{Preliminaries}\label{sec:Preliminary}
There are $n$ buyers competing for a single indivisible item.
Let $[n] = \{1,2, \cdots, n\}$ represent the set of buyers.
Buyer $i$ has a private valuation $v_i \in \mathbb{R}_{\ge 0}$, drawn independently from a (not necessarily identically distributed) prior distribution $F_i$ with density function $f_i$.
We use $\boldsymbol{v} = (v_1, v_2, \cdots, v_n)$ to denote the valuation profile of all $n$ buyers and $\boldsymbol{v}_{-i} = (v_1, \cdots, v_{i-1}, v_{i+1}, \cdots, v_n)$ to denote the valuation profile of all buyers except buyer $i$.
Each buyer $i$ submits a bid $b_i \in \mathbb{R}_{\ge 0}$, and the bid profile is given by $\boldsymbol{b} = (b_1, b_2, \cdots, b_n)$, with $\boldsymbol{b}_{-i} = (b_1, \cdots, b_{i-1}, b_{i+1}, \cdots, b_n)$ representing the bid profile of all buyers except buyer $i$.
Throughout, we restrict attention to truthful (incentive-compatible) mechanisms.
We use $x_i(\boldsymbol{v}): \mathbb{R}_{\ge 0}^n \rightarrow [0, 1]$, $p_i(\boldsymbol{v}): \mathbb{R}_{\ge 0}^n \rightarrow \mathbb{R}_{\ge0}$ to denote the ex-post allocation and the ex-post payment, respectively.
An auction is deterministic if $x_i(\boldsymbol{v}) \in \{0, 1\}$ for any $i$ and any $\boldsymbol{v}$.
With a slight abuse of notation, we use $x_i(v_i)$, $p_i(v_i)$ to denote interim allocation and interim payment, respectively.
The intended meaning will be clear from context.
$x_i(v_i)$ and $p_i(v_i)$ are defined as expectations over $\boldsymbol{v}_{-i}$.
Specifically, the interim allocation can be expressed as $x_i(v_i) = \int_{\boldsymbol{v}_{-i}} x_i(v_i, \boldsymbol{v}_{-i}) f_{-i}(\boldsymbol{v}_{-i}) \mathrm{d} \boldsymbol{v}_{-i}$ where $f_{-i}(\boldsymbol{v}_{-i}) = \prod_{j \neq i} f_j(v_j)$ and $\mathrm{d} \boldsymbol{v}_{-i} = \mathrm{d} v_1 \cdots \mathrm{d} v_{i - 1} \mathrm{d} v_{i+1} \cdots \mathrm{d} v_n$.
The interim payment can be expressed as $p_i(v_i) = \int_{\boldsymbol{v}_{-i}} p_i(v_i, \boldsymbol{v}_{-i}) f_{-i}(\boldsymbol{v}_{-i}) \mathrm{d} \boldsymbol{v}_{-i}$.

We consider the BIC, DSIC, and interim Individually Rational (IR) constraints in mechanism design.
A mechanism is BIC if truth-telling is the optimal strategy for each buyer, given that they have probabilistic beliefs about the valuations of other buyers.
A mechanism is DSIC if truth-telling always maximizes a buyer's utility, regardless of the actions of other buyers.
A mechanism is interim IR if each buyer receives a non-negative expected utility from participating.
Formally, these constraints are defined as follows:
\begin{itemize}
    \item BIC constraint: $v_i \cdot x_i(v_i) - p_i(v_i) \geq v_i \cdot x_i(b_i) - p_i(b_i)$, for all $i, v_i, b_i \in \mathbb{R}_{\geq 0}$;
    \item DSIC constraint: $v_i \cdot x_i(v_i, \boldsymbol{b}_{-i}) - p_i(v_i, \boldsymbol{b}_{-i}) \geq v_i \cdot x_i(b_i, \boldsymbol{b}_{-i}) -p_i(b_i, \boldsymbol{b}_{-i})$, for all $i$, $v_i$, $b_i \in \mathbb{R}_{\geq 0}$, and $\boldsymbol{b}_{-i} \in \mathbb{R}_{\geq 0}^{n - 1}$;
    \item Interim IR constraint: $v_i \cdot x_i(v_i) - p_i(v_i) \geq 0$, for all $i, v_i$.
\end{itemize}

According to \cite{myerson1981optimal}, a BIC auction is characterized by the following conditions: $x_i(v_i)$ is nondecreasing (monotone) in $v_i$ for all $i$, and the payment function satisfies the identity $p_i(v_i) = v_i \cdot x_i(v_i) - \int_0^{v_i} x_i(t) \mathrm{d} t + p_i(0)$.
Under interim IR, we normalize $p_i(0) = 0$.

According to \cite{nisan2007algorithmic}, a DSIC auction is characterized by the following conditions: $x_i(v_i, \boldsymbol{b}_{-i})$ is nondecreasing (monotone) in $v_i$ for all $i$ and any given $\boldsymbol{b}_{-i}$, and the payment function satisfies $v_i \cdot x_i(v_i, \boldsymbol{b}_{-i}) - \int_0^{v_i} x_i(t, \boldsymbol{b}_{-i}) \mathrm{d} t$.
Notably, any DSIC auction is also BIC.

The formulas for expected welfare and expected revenue are as follows:
\begin{align*}
    \textrm{WEL} &= \sum_{i=1}^n \int_{v_i} v_i x_i(v_i) f_i(v_i) \mathrm{d} v_i, \\
    \textrm{REV} &= \sum_{i=1}^n \int_{v_i} \phi_i(v_i) x_i(v_i) f_i(v_i) \mathrm{d} v_i,
\end{align*}
where $\phi_i(v_i) = v_i - \frac{1 - F_i(v_i)}{f_i(v_i)}$, known as the ``virtual value'' in \cite{myerson1981optimal}.
We also use the quantile space \cite{bulow1989simple}.
Let $q_i = F_i(v_i) \in [0, 1]$, and with a slight abuse of notation, let $v_i(q_i) = F_i^{-1}(q_i)$, the inverse function of $F_i(v_i)$.
Then, in the quantile space, we have 
\begin{align*}
    \textrm{WEL} &= \sum_{i=1}^n \int_{q_i} v_i(q_i) \hat{x}_i(q_i) \mathrm{d} q_i, \\
    \textrm{REV} &= \sum_{i=1}^n \int_{q_i} \hat{\phi}_i(q_i) \hat{x}_i(q_i) \mathrm{d} q_i,
\end{align*}
where $\hat{\phi}_i(q_i) = v_i(q_i) - (1 - q_i) v_i'(q_i)$, and $\hat{x}_i(q_i): [0, 1] \rightarrow [0, 1]$ is the interim allocation when $F_i(v_i) = q_i$ in the quantile space.

A point $P$ is a Pareto point of the range of the (revenue, welfare) pairs if there is no point $P' \neq P$ in the range that dominates $P$ (coordinate-wise).

\textbf{Implementability}: a (revenue, welfare) pair is implemented by mechanism $\mathcal{M}$ if and only if $\mathcal{M}$ yields revenue exactly and welfare exactly;
an interim allocation rule $\{x_i(\cdot)\}_{i \in [n]}$ is implemented by mechanism $\mathcal{M}$ if and only if $\mathcal{M}$ induces interim allocation exactly $\{x_i(\cdot)\}_{i \in [n]}$.

We present a variant of Border's theorem \cite{border1991implementation,che2013generalized} that, in our setting, characterizes exactly when an interim allocation is implemented by a randomized BIC auction.

\textbf{Border's theorem}: an interim allocation profile $(x_1(\cdot), \cdots, x_n(\cdot))$ can be implemented by a randomized BIC auction if and only if
\begin{align*}
    \forall (v_1, \cdots, v_n): \sum_i \int_{t_i = v_i}^{\infty} x_i(t_i) f_i(t_i) \mathrm{d} {t_i} \leq 1 - \prod_i F_i(v_i).
\end{align*}
In the i.i.d. symmetric case, in quantile space, it simplifies to
\begin{align*}
    \forall q \in [0, 1]: \int_{q}^1 \hat{x}(t) \mathrm{d} t \leq \frac{1 - q^n}{n},
\end{align*}
where $\hat{x}$ denotes the common interim allocation function.

The inequality compares two probabilities that must be consistent in any feasible single-item allocation.
Fix thresholds $(v_1, \cdots, v_n)$ and consider the event that every bidder $i$'s type exceeds $v_i$.
The left-hand side, $\sum_i \int_{t_i = v_i}^{\infty} x_i(t_i) f_i(t_i) \mathrm{d}t_i$, is the ex-ante probability that the mechanism allocates the item to some bidder from within the set of ``above-threshold'' types (counting bidder-by-bidder).
The right-hand side, $1 - \prod_i F_i(v_i)$, is the probability that at least one bidder is above her threshold.
Since only one item can be allocated, the probability of allocating the item to an above-threshold bidder can never exceed the probability that such a bidder exists.
Border's theorem states that this family of inequalities (for all thresholds) is not only necessary but also sufficient for interim implementability.
In our paper, this characterization lets us work directly with interim allocation rules in value/quantile space: we propose candidate interim allocations $\hat{x}$, verify feasibility by checking the corresponding Border constraints (e.g., the symmetric quantile form $\int_q^1 \hat{x}(t) \mathrm{d}t \le \frac{1 - q^n}{n}$), and then optimize linear objectives such as welfare and (virtual) revenue over this feasible set.
Finally, tightness of a Border constraint has a clean meaning: equality at $(v_1, \cdots, v_n)$ means there is no slack in the event ``someone is above threshold'', i.e., whenever at least one bidder exceeds her threshold, the mechanism allocates the item with probability 1 to an above-threshold bidder.
Strict inequality indicates that with positive probability the mechanism leaves allocation probability ``unused'' relative to that event (for example, it sometimes does not allocate to any above-threshold bidder even though one exists).
This interpretation is key for understanding the randomized benchmark region and for contrasting it with the additional structure imposed by deterministic DSIC implementability later in the paper.

\section{Implementation with Respect to (Revenue, Welfare) Pairs}\label{sec:pair}
\subsection{Nonequivalence of Deterministic BIC and Deterministic DSIC Auctions with Respect to Pareto Points in the Discrete Case}\label{sec:neq_pareto}
In Section \ref{sec:Preliminary}, values are drawn from continuous distributions.
Here we consider an i.i.d. discrete prior: each $v_i$ takes values in $V = \{ v^{1}, \cdots, v^{m} \}$ with probability vector $\boldsymbol{\pi} = (\pi_1, \cdots, \pi_m)$, where $\pi_k =  \mathrm{Pr}[v_i = v^{k}]$.
A deterministic mechanism $(x_i, p_i)$ specifies an ex post allocation $x_i(\mathbf{v}) \in \{0, 1\}$ and payment $p_i(\mathbf{v})$ for every profile $\mathbf{v} \in V^n$.
We evaluate it by expected welfare and expected revenue under the prior $\boldsymbol{\pi}$:
\begin{align*}
    \mathrm{WEL} &= \mathbb{E}_{\mathbf{v} \sim \boldsymbol{\pi}^{n}} \Big[\sum_{i \in [n]} v_i x_i(\mathbf{v}) \Big], \\
    \mathrm{REV} &= \mathbb{E}_{\mathbf{v} \sim \boldsymbol{\pi}^{n}} \Big[\sum_{i \in [n]} p_i(\mathbf{v}) \Big].
\end{align*}
Equivalently, $\mathbb{E}_{\mathbf{v} \sim \boldsymbol{\pi}^{n}}[\cdot]$ denotes a finite sum over all $\mathbf{v} = (v_1, \cdots, v_n) \in V^n$, where each coordinate $v_i = v^{k_i}$ is weighted by $\pi_{k_i}$.
Hence the profile $\mathbf{v}$ has weight $\prod_{i \in [n]} \pi_{k_i}$.

We give an example in the i.i.d. discrete case to derive the following theorem.
\begin{theorem}\label{Thm:Nonequivalence-discrete}
    There exists a Pareto point $P$ that is implemented by a deterministic BIC auction but cannot be implemented by any deterministic DSIC auction in the i.i.d. discrete case.
\end{theorem}

Suppose that there are two identical buyers, each with three types: $t_1 = 100, t_2 = 10, t_3 = 1$, occurring with probabilities $p_1 = 0.8, p_2 = 0.15, p_3 = 0.05$ respectively.
For each type of buyer 1 and buyer 2, the item can be allocated to no one, to buyer 1, or to buyer 2.
Thus, we have to check $3^9 = 19683$ deterministic allocation profiles.
Similar to \cite{diakonikolas2012efficiency}, to consider Pareto-optimal points, we can assume without loss of generality that payments are specified on the finite type space.
Therefore, we can enumerate all deterministic allocation rules on the finite type space and, for each, compute supporting payments and check BIC/DSIC/IR \footnote{Otherwise, there would be infinitely many feasible actions.}.
Through experiments, we find that $(\textrm{WEL} = 96.2275, \textrm{REV} = 85)$ is a Pareto-optimal point among all outcomes implementable by deterministic BIC mechanisms (see Figure \ref{fig:counter1}).
Only partial figure is provided here since the complete figure is too large.
\begin{figure}[htbp]
    \centering
    \includegraphics[width=0.45\textwidth]{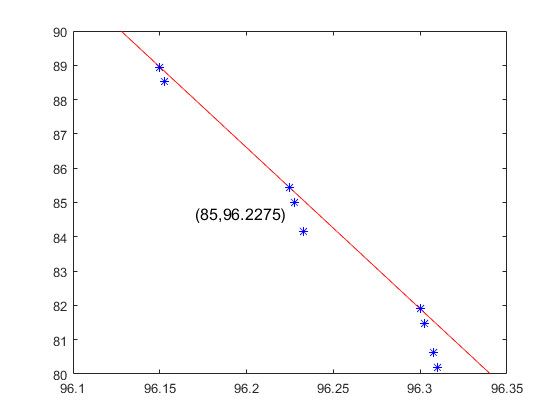}
    \caption{Pareto points in the area ($\textrm{WEL}>96$ and $\textrm{REV}>80$). Each blue point represents a Pareto point. The red line characterizes the boundary. The horizontal axis represents welfare, and the vertical axis represents revenue.}
    \label{fig:counter1}
\end{figure}
The Pareto-optimal point $(\textrm{WEL} = 96.2275, \textrm{REV} = 85)$ can be achieved only by four deterministic BIC auctions, with $i \in \{1, 2\}$ and $j \in \{1, 2\}$, as shown in Table \ref{Tab:4-DIC auctions}.
However, none of these auctions is DSIC.

\begin{table}[htbp]
    \caption{Four deterministic BIC auctions implement the Pareto-optimal point $(\textrm{WEL} = 96.2275, \textrm{REV} = 85)$. Here the identification $v_1(v_2)$ represents the type of buyer 1(2). The elements represent the buyer who gets the item, where ``seller'' means that no buyer gets the item.}
    \centering
    \label{Tab:4-DIC auctions}
    \begin{tabular} {|c|c|c|c|} \hline
        \diagbox{$v_2$}{$v_1$}	& $t_3 = 1$ & $t_2 = 10$ & $t_1 = 100$ \\
        \hline
        $t_3 = 1$ & buyer $j$ & seller & buyer 1 \\
        \hline
        $t_2 = 10$ & seller & buyer $j$ & buyer 1 \\
        \hline
        $t_1 = 100$ & buyer 2 & buyer 2 & buyer $i$ \\
        \hline
    \end{tabular}
\end{table}
Note that in Table \ref{Tab:4-DIC auctions}, buyer $j$ and buyer $i$ denote any buyer.
Therefore, the four deterministic auctions are: (1) $i = 1, j = 1$; (2) $i = 1, j = 2$; (3) $i = 2, j = 1$; (4) $i = 2, j = 2$.
Here, we use the case $i = 2, j = 1$ to illustrate that this deterministic auction is BIC but not DSIC.
For buyer 1, we can get the interim allocation, that is $x_1(1) = 0.05$, $x_1(10) = 0.15$ and $x_1(100) = 0.2$.
For buyer 2, we can get the interim allocation, that is $x_2(1) = 0$, $x_2(10) = 0$ and $x_2(100) = 1$.
Clearly, the allocation functions of both buyers are monotonic.
The welfare generated by buyer 1 is $0.05 \times 1 \times 0.05 + 0.15 \times 10 \times 0.15 + 0.2 \times 100 \times 0.8 = 16.2275$; the welfare generated by buyer 2 is $0 \times 1 \times 0.05 + 0 \times 10 \times 0.15 + 1 \times 100 \times 0.8 = 80$; the total welfare is $16.2275 + 80 = 96.2275$.
As for revenue, we have to design payment rules for both buyers.
However, unlike the case with continuous types, when types are discrete, there can be multiple possible payment rules that satisfy BIC.
Based on the established allocation rule, we can formulate an optimization problem that maximizes the seller's expected revenue while incorporating the BIC constraints as part of the feasible set.
By solving this constrained optimization problem, we can derive the corresponding payment functions for both buyers.
Finally, we can get the interim payment, that is $p_1(1) = 0.05$, $p_1(10) = 1.05$, $p_1(100) = 6.05$, $p_2(1) = 0$, $p_2(10) = 0$, $p_2(100) = 100$.
The revenue generated by buyer 1 is $0.05 \times 0.05 + 1.05 \times 0.15 + 6.05 \times 0.8 = 5$; the revenue generated by buyer 2 is $0 \times 0.05 + 0 \times 0.15 + 100 \times 0.8 = 80$; the total revenue is $5 + 80 = 85$.

Although the interim allocation rules are monotone (hence the mechanism is BIC), the ex-post allocation rule is not monotone in a bidder's own report for fixed opponent reports, so it cannot be DSIC.
For instance, fixing buyer 2's report at $v_2 = 1$, buyer 1 wins when $v_1 = 1$, loses when $v_1 = 10$, and wins again when $v_1 = 100$, which violates the monotonicity requirement for deterministic DSIC implementability.
Hence, none of these four auctions can satisfy the DSIC constraints, even though they all achieve the same Pareto-optimal (revenue, welfare) pair under BIC.

\subsection{Equivalence of Deterministic DSIC and Randomized BIC Auctions with Respect to Pareto Points for Regular Distribution}\label{sec:eq_pareto}
In this subsection, we assume that the buyers' prior distributions are continuous and atomless, but not necessarily identical.
\begin{theorem}\label{Thm:Pareto_regular}
    If each buyer's prior distribution is regular, i.e., the virtual value $\phi(v) = v - \frac{1 - F(v)}{f(v)}$ is increasing, then any Pareto point implementable by a randomized BIC auction is also implementable by a deterministic DSIC auction.
\end{theorem}
We prove the theorem in two steps.
In the first step, starting from a randomized BIC mechanism, we transform its allocation rule by modifying only one buyer's interim allocation at a time, while keeping the other buyers' interim allocations fixed.
We repeat this procedure for all $n$ buyers.
Under the regularity condition, each transformation weakly improves both expected welfare and expected revenue, so the final mechanism weakly dominates the original one in the $(\mathrm{REV}, \mathrm{WEL})$ space.
In the second step, we show that the last (transformed) allocation rule admits a deterministic DSIC implementation.
Since the original point is already Pareto-optimal among outcomes achievable by randomized BIC mechanisms, the weak improvement must be tight.
Therefore, the deterministic DSIC mechanism implements the same Pareto point.
The full proof of Theorem \ref{Thm:Pareto_regular} is deferred to Appendix \ref{Proof-Thm:Pareto_regular}.

\subsection{Equivalence of Deterministic DSIC and Randomized BIC Auctions with Respect to All (Revenue, Welfare) Pairs for i.i.d Distribution}\label{sec:eqpair}
In this subsection, we focus on the i.i.d case with a continuous and atomless distribution $F$ and utilize the quantile space.
Since buyers' private valuations are independently drawn from a common distribution $F$, we omit the subscripts on $v_i$ and $q_i$ and write them simply as $v$ and $q$.
We provide sufficient conditions under which the range of (revenue, welfare) pairs achievable by deterministic auctions is equal to that of randomized BIC auctions.
\begin{theorem}\label{thm:equal_to_random}
    If the prior distribution $F$ is third-order differentiable and satisfies the following condition, then any (revenue, welfare) pair that is implementable by a randomized auction is also implementable by a deterministic DSIC auction.
    \begin{equation}\label{thrid_order}
        v''(q) \geq 0, \quad \hat{\phi}''(q) \leq 0, \quad \forall q.
    \end{equation}
\end{theorem}
It is not hard to verify that all uniform distributions and power-law distributions $F(v) = 1 - v^{-\alpha}$ with $\alpha \leq 1$, including the equal-revenue distributions, satisfy the condition.

The proof sketch is as follows: we first characterize the boundary of randomized BIC auctions, then we prove that each (revenue, welfare) pair on the boundary can be implemented by a deterministic auction (see Lemmas \ref{lem:solutions_of_o2} - \ref{Lem:sum_equal}).
Finally we prove that each (revenue, welfare) pair strictly between the revenue-maximizing boundary and the revenue-minimizing boundary at the same welfare level can be implemented by a deterministic auction (see Lemma \ref{lem:inner-pairs}).
The full proof of Theorem \ref{thm:equal_to_random} is deferred to  Appendix \ref{Proof:equal_to_random}.

We begin by analyzing the boundary of (revenue, welfare) pairs for randomized BIC auctions.
For convenience, we assume that $F$ has full support on $[0, 1]$.
By \cite{maskin1984optimal}, it is without loss of generality to only consider symmetric auctions.
Define $P(q) = \int_q^1 \hat{x}(t) \mathrm{d} t$ and we have
\begin{align*}
    \textrm{WEL} &= n \int_0^1 v(q) \hat{x}(q) \mathrm{d} q = -n \int_0^1 v(q) \mathrm{d} P(q) \\
    &= n \int_0^1 v'(q) P(q) \mathrm{d} q, \\
    \textrm{REV} &= n \int_0^1 \hat{\phi}(q) \hat{x}(q) \mathrm{d} q= -n \int_0^1 \hat{\phi}(q) \mathrm{d} P(q) \\
    &= n \hat{\phi}(0) P(0) + n \int_0^1 \hat{\phi}'(q) P(q) \mathrm{d} q.
\end{align*}

Note that $P'(q) = -\hat{x}(q) \leq 0$ and $P''(q) = -\hat{x}'(q) \leq 0$ for a randomized BIC auction.
By the feasibility from Border's theorem, $P(q)$ is bounded by $P(q) \leq \frac{1 - q^n}{n}$.
We fix $\textrm{WEL} = c$ where $c$ is a constant, and then compute the range of $\textrm{REV}$.
Thus, the optimization problem is
\begin{align}\label{Objective2}
    \max / \min ~~~~ & n \hat{\phi}(0) P(0) + n \int_0^1 \hat{\phi}'(q) P(q) \mathrm{d} q \\ \nonumber
    \text{s.t.} ~~~~ & n \int_0^1 v'(q) P(q) \mathrm{d} q = c, ~ P(q) \leq \frac{1 - q^n}{n}, \\ \nonumber
    ~~~~ & P'(q) \leq 0, ~ P''(q) \leq 0, ~ P(1) = 0. \nonumber
\end{align}

We define a class of auctions called ``piecewise auctions''.
We show that any point on the boundary of the feasible range can be implemented by a piecewise auction.

\begin{definition}[Piecewise auctions]\label{def:piecewise auctions}
    A (symmetric, randomized) BIC auction is a piecewise auction if it satisfies one of the following two conditions:\\
    1. There exists $0 \leq r_1 < r_2 \leq 1$ such that
	\begin{itemize}
		  \item $P'(q) = 0$, for $0 < q \le r_1$,
		  \item $P'(q) = -k$, for $r_1 < q \le r_2$, where $k$ is a constant and $0 \leq k \leq \frac{r_2^n - r_1^n}{n(r_2 - r_1)}$,
		  \item $P(q) = \frac{1 - q^n}{n}$, for $r_2 < q \leq 1$.
	\end{itemize}
    2. There exists $0 \leq r_1 < r_2 \leq 1$ such that
	\begin{itemize}
		  \item $P'(q) = 0$, for $0 < q \le r_1$,
		  \item $P(q) = \frac{1 - q^n}{n}$, for $r_1 < q \leq r_2$,
		  \item $P'(q) = \frac{1 - r_2^n}{n(r_2 - 1)}$, for $r_2 < q \leq 1$.
	\end{itemize}
\end{definition}

According to the piecewise auctions we defined, we introduce the following Lemmas.
Lemma \ref{lem:solutions_of_o2} shows that any maximum and any minimum solutions of optimization problem \eqref{Objective2} must be a piecewise auction, so all boundary points of the $(\mathrm{REV}, \mathrm{WEL})$ region induced by randomized BIC auctions are attained by piecewise auctions.
Lemma \ref{lem:implementation_of_simple_auctions} then constructs, for each piecewise auction, a deterministic DSIC mechanism whose total interim allocation at every valuation coincides with that of the piecewise auction; together with Lemma \ref{Lem:sum_equal}, which states that identical total interim allocations imply identical revenue and welfare, this implies that every boundary point can be implemented by a deterministic DSIC auction.

\begin{lemma}\label{lem:solutions_of_o2}
    If the condition (\ref{thrid_order}) is satisfied, then any maximum and any minimum solutions of optimization problem (\ref{Objective2}) are a piecewise auction.
\end{lemma}

\begin{lemma}\label{Lem:sum_equal}
    For any common prior distribution, if two auctions with interim allocation profile $(x_1(v), x_2(v), \cdots, x_n(v))$ and $(y_1(v), y_2(v), \cdots, y_n(v))$ respectively satisfy for any $v$,
    \begin{align*}
        \sum_{i = 1}^n x_i(v) = \sum_{i = 1}^n y_i(v),
    \end{align*}
    then the two auctions' revenue and welfare are the same.
\end{lemma}

By Lemma \ref{Lem:sum_equal}, we can get Lemma \ref{lem:implementation_of_simple_auctions}.
\begin{lemma}\label{lem:implementation_of_simple_auctions}
    All (revenue, welfare) pairs of piecewise auctions can be implemented by deterministic DSIC auctions.
\end{lemma}

Lemma \ref{lem:inner-pairs} further shows that, for any fixed welfare level, piecewise auctions span a continuous interval in the revenue dimension, so all interior points on the same welfare line are also implemented by piecewise auctions (and hence by deterministic DSIC mechanisms).
Therefore, the entire $(\mathrm{REV}, \mathrm{WEL})$ region attainable by randomized auctions can be implemented by deterministic DSIC auctions, which completes the proof of Theorem \ref{thm:equal_to_random}.
\begin{lemma}\label{lem:inner-pairs}
    One piecewise auction can continuously transfer to another through a sequence of piecewise auctions with the welfare fixed.
\end{lemma}

\section{Implementations with Respect to Interim Allocations}\label{sec:interim}
\subsection{Two Buyers}\label{sec:2buyer}
In this subsection, we analyze the necessary and sufficient condition such that an interim allocation profile can be implemented by a deterministic DSIC auction, also known as Border's Theorem in the deterministic case.
We focus on the two-buyer case, with continuous and atomless prior distributions, which are not necessarily symmetric.
The analysis is conducted in the quantile space.
\begin{theorem}\label{Thm:Border_deterministic}
    A strictly increasing interim allocation profile $(\hat{x}_1(q_1), \hat{x}_2(q_2))$ is implementable by a deterministic DSIC auction if and only if
    \begin{align}\label{Border_deterministic}
        \forall q_1 \in [0, 1], ~\hat{x}_2(\hat{x}_1(q_1)) \leq q_1.
    \end{align}
    Moreover, if the item is always sold for the deterministic DSIC auction, then (\ref{Border_deterministic}) turns into
    \begin{align}\label{Border_deterministic2}
		\forall q_1, ~\hat{x}_2(\hat{x}_1(q_1)) = q_1.
    \end{align}
\end{theorem}
The proof sketch is as follows: we work on the quantile space $[0, 1]^2$ and represent any deterministic allocation rule as a three-coloring $C \in \{0, 1, 2\}$, where 0, 1, and 2 indicate no sale, buyer 1 wins, and buyer 2 wins, respectively.
We then perform a measure-preserving rearrangement to obtain a new coloring $C'$ that keeps each buyer's interim allocation unchanged.
Concretely, for each $q_1$, the measure of points colored 1 in column $q_1$ remains $\hat{x}_1(q_1)$, and for each $q_2$, the measure of points colored 2 in row $q_2$ remains $\hat{x}_2(q_2)$.
This rearrangement produces a structured geometry where the plane is separated by two monotone boundary curves $q_2 = \hat{x}_1(q_1)$ and $q_1 = \hat{x}_2(q_2)$, and each buyer's winning set has a threshold form.
In this form, monotonicity in each buyer's own quantile is immediate, so the induced allocation can be implemented by a deterministic DSIC auction via the standard critical-threshold characterization.
Finally, single-item feasibility requires the two winning regions to be disjoint, which is equivalent to the two monotone boundaries not crossing, and this gives $\forall q$, $\hat{x}_2(\hat{x}_1(q))\le q$.
The full proof of Theorem \ref{Thm:Border_deterministic} is deferred to Appendix \ref{Proof-Border_deterministic}.


\subsubsection*{Nonequivalence of Randomized BIC and Deterministic DSIC Auctions with Respect to Interim Allocations}
Based on Theorem \ref{Thm:Border_deterministic}, we can prove the nonequivalence between randomized BIC and deterministic DSIC auctions with respect to interim allocations.
A straightforward approach is to consider the sum of interim allocations given a valuation $v$.
\begin{corollary}\label{lem:corollary}
    There exists an interim allocation profile $(x_1(v), \cdots, x_n(v))$ that is implementable by a randomized BIC auction, but there is no deterministic DSIC auction with interim allocation profile $(y_1(v), \cdots, y_n(v))$ such that for any $v$,
    \begin{align*}
        \sum_{i = 1}^n x_i(v) = \sum_{i = 1}^n y_i(v).
    \end{align*}
\end{corollary}
We outline the idea of the proof.
We consider the two-buyer case and analyze the following randomized BIC auction:
\begin{itemize}
    \item with probability $p$, run VCG auction.
    \item with probability $1 - p$, send the item for free to buyer 1.
\end{itemize}

We know that the item is always sold. 
By Theorem ~\ref{Thm:Border_deterministic}, we have $\hat{x}_2(\hat{x}_1(q_1)) = q_1$ for all $q_1$.
This allows us to derive a series of interim allocations under a deterministic DSIC auction.
However, the interim allocation may exceed 1, thereby proving the corollary.

\subsection{Three Buyers with Constant Interim Allocations}\label{sec:3buyer}
Theorem \ref{Thm:Border_deterministic} becomes extremely complex when extended to cases with three or more buyers.
Suppose there are three buyers with constant interim allocations, i.e. $x_i(v) = c_i, \forall v, i = 1, 2, 3$.
We provide a necessary and sufficient condition for interim allocations to be implemented by a deterministic DSIC auction.
\begin{theorem}\label{Thm:three_buyer_border}
    The constant interim allocation profile for 3 buyers is implementable by a deterministic DSIC auction if and only if 
    \begin{eqnarray}\label{three_buyers_condition}
        1 - c_1 - c_2 - c_3 \geq 2 \sqrt{c_1 c_2 c_3}.
    \end{eqnarray}
\end{theorem}
\begin{proof}
    We first prove the ``only if'' direction: given constant interim allocation rule $x_1, x_2, x_3$ in a deterministic DSIC auction, then condition (\ref{three_buyers_condition}) holds.
    
    Due to the property of deterministic DSIC and and constant interim allocation rule, we have the following statement:
    \vskip 5pt
    \noindent \emph{Except for a zero measure of $(b, c) \in [0, 1]^2$, either $x_1(v, b, c) = 0, \forall v$ or $x_1(v, b, c) = 1, \forall v$.}
    \vskip 5pt
    It means that a buyer's allocation rule does not depend on his own value.
    In the following proof we ignore the unsatisfactory set which is measured 0 and has no effect to results. 
    
    Define three valuation domain sets: $A = \{(b, c): x_1(v_1, b, c) = 1\}$ for $\forall v_1$, $B = \{(a, c): x_2(a, v_2, c) = 1\}$ for $\forall v_2$ and $C = \{(a, b): x_3(a, b, v_3) = 1 \}$ for $\forall v_3$.
    Each set represents two players' valuation profile when the mechanism allocates the item to the remaining player.
    
    Define the projection of set A in each dimension as $A_b = \{b: \exists c ~s.t. ~(b, c) \in A\}$ and $A_c = \{c: \exists b ~s.t. ~(b, c) \in A\}$.
    Define $B_a, B_c, C_a, C_b$ similarly.
    
    We next claim that $B_a \cap C_a = \emptyset$, $A_b \cap C_b = \emptyset$, $A_c \cap B_c = \emptyset$.
    Suppose $B_a \cap C_a \neq \emptyset$, let $a^*$ be an element in $B_a \cap C_a$.
    Assume $(a^*, c^*) \in B$ and $(a^*, b^*) \in C$.
    By definition of $B$,  $x_2(a^*, b^*, c^*) = 1$, while by definition of $C$, $x_3(a^*, b^*, c^*) = 1$ which leads a contradiction.
    So $B_a$ and $C_a$ do not overlap.
    Same argument applies for $A_b$ and $C_b$, $A_c$ and $B_c$.
    
    Define $y_1 = |A_b|$, $y_2 = |B_c|$, and $y_3 = |C_a|$.
    Then we have $|C_b| \leq 1 - y_1$, $|A_c| \leq 1 - y_2$ and $|B_a| \leq 1 - y_3$ (since $A_b \cap C_b = \emptyset$).
    
    The interim allocation $c_1 = |A|$ is bounded by $|A_b||A_c| \leq y_1(1 - y_2)$, since $A_b$ and $A_c$ are projections of $A$.
    
    Thus, the interim allocation profile $c_1, c_2, c_3$ is implementable implies that there exists $y_1, y_2, y_3 \in [0, 1]$ such that
    \begin{eqnarray*}
        c_1 \leq y_1(1 - y_2), \\
        c_2 \leq y_2(1 - y_3), \\
        c_3 \leq y_3(1 - y_1).
    \end{eqnarray*}
    
    We have
    \[c_2 \leq y_2(1 - y_3) \leq (1 - \frac{c_1}{y_1})(1 - \frac{c_3}{1 - y_1}).\]
    That is,
    \[(1 - c_2) y_1^2 + (c_2 - c_1 + c_3 - 1) y_1 + c_1(1 - c_3) \leq 0.\]
    
    Due to the formula above has a solution for $y_1$, we have
    \begin{align*}
        (c_2 - c_1 + c_3 - 1) \geq 4c_1(1 - c_2)(1 - c_3),
    \end{align*}
    That is,
    \begin{align*}
        (c_1 + c_2 + c_3 - 1)^2 \geq 4c_1 c_2 c_3.
    \end{align*}
    Note that
    \begin{align*}
        1 & \geq \mathbb{E}_{v_1, v_2, v_3}[x_1(v_1, v_2, v_3) + x_2(v_1, v_2, v_3) + x_3(v_1, v_2, v_3)] \\
        & = \mathbb{E}_{v_1}[x_1(v_1)] + \mathbb{E}_{v_2}[x_2(v_2)] + \mathbb{E}_{v_3}[x_3(v_3)] \\
        & = c_1 + c_2 + c_3.
    \end{align*}
    Consequently, we have
    \begin{align*}
        1 - c_1 - c_2 - c_3 \geq 2 \sqrt{c_1 c_2 c_3}.
    \end{align*}
    
    For the ``if'' direction: define $\Delta = (c_1 + c_2 + c_3 - 1)^2 - 4c_1 c_2 c_3$ and suppose $\Delta \geq 0$.
    We choose $y_1 = \frac{1 + c_1 - c_2 - c_3 + \sqrt{\Delta}}{2(1 - c_2)}$,
    $y_2 = \frac{1 + c_2 - c_1 - c_3 + \sqrt{\Delta}}{2(1 - c_3)}$,
    $y_3 = \frac{1 + c_3 - c_2 - c_1 + \sqrt{\Delta}}{2(1 - c_1)}$.
    
    Then, we construct an auction with the following allocation rule, which is deterministic DSIC.
    \begin{itemize}
        \item $x_1(v_1, v_2, v_3) = 1$ if and only if $v_2 \leq y_1$ and $v_3 > y_2$.
        \item $x_2(v_1, v_2, v_3) = 1$ if and only if $v_3 \leq y_2$ and $v_1 > y_3$.
        \item $x_3(v_1, v_2, v_3) = 1$ if and only if $v_1 \leq y_3$ and $v_2 > y_1$.
    \end{itemize}
    
    Finally, we have to prove that ``Except for a zero measure of $(b, c) \in [0, 1]^2$, either $x_1(v, b, c) = 0, \forall v$ or $x_1(v, b, c) = 1, \forall v$'':
    
    According to the definition of deterministic DSIC auction, given $(b, c)$, there exist a critical value $cr(b, c)$ such that
    \begin{itemize}
        \item $x_1(v, b, c) = 0$, if $v < cr(b, c)$,
        \item $x_1(v, b, c) = 1$, if $v > cr(b, c)$,
        \item $x_1(v, b, c) = 0$ or $x_1(v, b, c) = 1$, if $v = cr(b, c)$.
    \end{itemize}
    Define $A_1 = \{(b, c)| cr(b, c) = 0 \}$, $A_2 = \{(b, c) | cr(b, c) < 1\}$.
    Clearly $A_1 \subseteq A_2$.
    Note that
    \begin{align*}
        c_1 & = x_1(0) =  E_{b, c}[x_1(0, b, c)] \leq |A_1| \leq |A_2| \\
        & \leq E_{b, c}[x_1(1, b, c)] = x_1(1) = c_1.
    \end{align*}
    Hence $|A_1| = |A_2|$, and since $A_1 \subseteq A_2$, we must have $|A_2 \backslash A_1| = 0$.
    But $A_2 \backslash A_1 = \{(b, c): 0 < cr(b, c) < 1\}$.
    Therefore, except for a zero-measure set of $(b, c)$, we have $cr(b, c) \in \{0, 1\}$, which implies $x_1(v, b, c)$ is constant in $v$ (up to the boundary point $v = 0$ or $v = 1$, which does not affect interim allocations).
    This completes the proof.
\end{proof}

\section{Discussion}\label{conclusion}
We study the implementation power of deterministic single-item auctions from two complementary perspectives: attainable (revenue, welfare) trade-offs and implementable interim allocation rules.
Our results delineate a clear gap between deterministic BIC and deterministic DSIC, while also identifying distributional regimes where deterministic DSIC can match the outcomes of randomized benchmarks.
Two natural directions follow.
First, it remains open how far the current distributional assumptions can be weakened, and whether broader families of priors admit the same outcome equivalence.
Second, our interim-allocation characterization is currently limited to two bidders (and a restricted three-bidder case).
Extending it to general $n$, even under structured restrictions such as constant or low-complexity reduced forms, would substantially sharpen the boundary of deterministic implementability.

\clearpage
\bibliographystyle{named} 
\bibliography{reference}

@inproceedings{likhodedov2003mechanism,
  title={Mechanism for optimally trading off revenue and efficiency in multi-unit auctions},
  author={Likhodedov, Anton and Sandholm, Tuomas},
  booktitle={International Workshop on Agent-Mediated Electronic Commerce},
  pages={92--108},
  year={2003},
  organization={Springer}
}

@inproceedings{diakonikolas2012efficiency,
  title={Efficiency-revenue trade-offs in auctions},
  author={Diakonikolas, Ilias and Papadimitriou, Christos and Pierrakos, George and Singer, Yaron},
  booktitle={International Colloquium on Automata, Languages, and Programming},
  pages={488--499},
  year={2012},
  organization={Springer}
}

@article{che2013generalized,
  title={Generalized Reduced-Form Auctions: A Network-Flow Approach},
  author={Che, Yeon-Koo and Kim, Jinwoo and Mierendorff, Konrad},
  journal={Econometrica},
  volume={81},
  number={6},
  pages={2487--2520},
  year={2013},
  publisher={Wiley Online Library}
}

@inproceedings{hartline2008optimal,
  title={Optimal mechanism design and money burning},
  author={Hartline, Jason D and Roughgarden, Tim},
  booktitle={Proceedings of the fortieth annual ACM symposium on Theory of computing},
  pages={75--84},
  year={2008},
  organization={ACM}
}

@article{bergemann2015limits,
  title={The limits of price discrimination},
  author={Bergemann, Dirk and Brooks, Benjamin and Morris, Stephen},
  journal={The American Economic Review},
  volume={105},
  number={3},
  pages={921--957},
  year={2015},
  publisher={American Economic Association}
}

@article{myerson1981optimal,
  title={Optimal auction design},
  author={Myerson, Roger B},
  journal={Mathematics of operations research},
  volume={6},
  number={1},
  pages={58--73},
  year={1981},
  publisher={INFORMS}
}

@article{bikhchandani2006weak,
  title={Weak monotonicity characterizes deterministic dominant-strategy implementation},
  author={Bikhchandani, Sushil and Chatterji, Shurojit and Lavi, Ron and Mu'alem, Ahuva and Nisan, Noam and Sen, Arunava},
  journal={Econometrica},
  volume={74},
  number={4},
  pages={1109--1132},
  year={2006},
  publisher={Wiley Online Library}
}

@article{manelli2010bayesian,
  title={Bayesian and Dominant-Strategy Implementation in the Independent Private-Values Model},
  author={Manelli, Alejandro M and Vincent, Daniel R},
  journal={Econometrica},
  volume={78},
  number={6},
  pages={1905--1938},
  year={2010},
  publisher={Wiley Online Library}
}

@article{gershkov2013equivalence,
  title={On the equivalence of Bayesian and dominant strategy implementation},
  author={Gershkov, Alex and Goeree, Jacob K and Kushnir, Alexey and Moldovanu, Benny and Shi, Xianwen},
  journal={Econometrica},
  volume={81},
  number={1},
  pages={197--220},
  year={2013},
  publisher={Wiley Online Library}
}

@article{maskin1984optimal,
  title={Optimal auctions with risk averse buyers},
  author={Maskin, Eric and Riley, John},
  journal={Econometrica: Journal of the Econometric Society},
  pages={1473--1518},
  year={1984},
  publisher={JSTOR}
}

@article{bulow1989simple,
  title={The simple economics of optimal auctions},
  author={Bulow, Jeremy and Roberts, John},
  journal={The Journal of Political Economy},
  pages={1060--1090},
  year={1989},
  publisher={JSTOR}
}

@article{myerson1983efficient,
  title={Efficient mechanisms for bilateral trading},
  author={Myerson, Roger B and Satterthwaite, Mark A},
  journal={Journal of economic theory},
  volume={29},
  number={2},
  pages={265--281},
  year={1983},
  publisher={Elsevier}
}

@article{border1991implementation,
  title={Implementation of reduced form auctions: A geometric approach},
  author={Border, Kim C},
  journal={Econometrica: Journal of the Econometric Society},
  pages={1175--1187},
  year={1991},
  publisher={JSTOR}
}

@article{pai2009competing,
  title={Competing auctioneers},
  author={Pai, Mallesh},
  journal={Discussion papers, Northwestern University},
  year={2009},
  publisher={Citeseer}
}

@article{border2007reduced,
  title={Reduced form auctions revisited},
  author={Border, Kim C},
  journal={Economic Theory},
  volume={31},
  number={1},
  pages={167--181},
  year={2007},
  publisher={Springer}
}

@inproceedings{kleinberg2013ratio,
  title={On the ratio of revenue to welfare in single-parameter mechanism design},
  author={Kleinberg, Robert and Yuan, Yang},
  booktitle={Proceedings of the fourteenth ACM conference on Electronic commerce},
  pages={589--602},
  year={2013},
  organization={ACM}
}

@book{nisan2007algorithmic,
  title={Algorithmic game theory},
  author={Nisan, Noam and Roughgarden, Tim and Tardos, Eva and Vazirani, Vijay V},
  volume={1},
  year={2007},
  publisher={Cambridge University Press Cambridge}
}

@article{klemperer2004auctions,
  title = {Auctions: Theory and Practice},
  author = {Klemperer, Paul},
  year = {2004},
  journal = {SSRN Electronic Journal},
  doi = {10.2139/SSRN.491563},
  url = {https://doi.org/10.2139/SSRN.491563}
}

@book{hurwicz1960optimality,
  title={Optimality and informational efficiency in resource allocation processes},
  author={Hurwicz, Leonid},
  year={1960},
  publisher={Stanford University Press Stanford, CA}
}

@article{hurwicz1972informationally,
  title={On informationally decentralized systems},
  author={Hurwicz, Leonid},
  journal={Decision and organization},
  year={1972},
  publisher={North Holland, Amsterdam, The Netherlands}
}

@article{vickrey1961counterspeculation,
  title={Counterspeculation, auctions, and competitive sealed tenders},
  author={Vickrey, William},
  journal={The Journal of finance},
  volume={16},
  number={1},
  pages={8--37},
  year={1961},
  publisher={Wiley Online Library}
}

@article{clarke1971multipart,
  title={Multipart pricing of public goods},
  author={Clarke, Edward H},
  journal={Public choice},
  volume={11},
  number={1},
  pages={17--33},
  year={1971},
  publisher={Springer}
}

@article{groves1973incentives,
  title={Incentives in teams},
  author={Groves, Theodore},
  journal={Econometrica: Journal of the Econometric Society},
  pages={617--631},
  year={1973},
  publisher={JSTOR}
}

@article{myerson1979incentive,
  title={Incentive compatibility and the bargaining problem},
  author={Myerson, Roger B},
  journal={Econometrica: journal of the Econometric Society},
  pages={61--73},
  year={1979},
  publisher={JSTOR}
}

@article{myerson1982optimal,
  title={Optimal coordination mechanisms in generalized principal--agent problems},
  author={Myerson, Roger B},
  journal={Journal of mathematical economics},
  volume={10},
  number={1},
  pages={67--81},
  year={1982},
  publisher={Elsevier}
}

@inproceedings{tang2016optimal,
  title={Optimal auctions for negatively correlated items},
  author={Tang, Pingzhong and Wang, Zihe},
  booktitle={Proceedings of the 2016 ACM Conference on Economics and Computation},
  pages={103--120},
  year={2016},
  organization={ACM}
}

@inproceedings{yao2017dominant,
  title={Dominant-strategy versus bayesian multi-item auctions: Maximum revenue determination and comparison},
  author={Yao, Andrew Chi-Chih},
  booktitle={Proceedings of the 2017 ACM Conference on Economics and Computation},
  pages={3--20},
  year={2017},
  organization={ACM}
}

@article{arigapudi2018equivalence,
  title={On the equivalence of Bayesian and deterministic dominant strategy implementation},
  author={Arigapudi, Srinivas},
  journal={Economics Letters},
  volume={162},
  pages={37--40},
  year={2018},
  publisher={Elsevier}
}

@inproceedings{shen2017practical,
  title={Practical versus optimal mechanisms},
  author={Shen, Weiran and Tang, Pingzhong},
  booktitle={Proceedings of the 16th Conference on Autonomous Agents and MultiAgent Systems},
  pages={78--86},
  year={2017},
  organization={International Foundation for Autonomous Agents and Multiagent Systems}
}

@inproceedings{bachrach2014optimising,
  title={Optimising trade-offs among stakeholders in ad auctions},
  author={Bachrach, Yoram and Ceppi, Sofia and Kash, Ian A and Key, Peter and Kurokawa, David},
  booktitle={Proceedings of the fifteenth ACM conference on Economics and computation},
  pages={75--92},
  year={2014},
  organization={ACM}
}

@article{chen2019equivalence,
  title={Equivalence of Stochastic and Deterministic Mechanisms},
  author={Chen, Yi-Chun and He, Wei and Li, Jiangtao and Sun, Yeneng},
  journal={Econometrica},
  volume={87},
  number={4},
  pages={1367--1390},
  year={2019},
  publisher={Wiley Online Library}
}

@article{ma2022revenue,
  title={Revenue-optimal deterministic auctions for multiple buyers with ordinal preferences over fixed-price items},
  author={Ma, Will},
  journal={ACM Transactions on Economics and Computation},
  volume={10},
  number={2},
  pages={1--32},
  year={2022},
  publisher={ACM New York, NY}
}

@article{you2021incentive,
  title={Incentive-compatible simple mechanisms},
  author={You, Jung S and Juarez, Ruben},
  journal={Economic Theory},
  volume={71},
  number={4},
  pages={1569--1589},
  year={2021},
  publisher={Springer}
}

@inproceedings{dobzinski2022hardness,
  title={On the hardness of dominant strategy mechanism design},
  author={Dobzinski, Shahar and Ron, Shiri and Vondr{\'a}k, Jan},
  booktitle={Proceedings of the 54th Annual ACM SIGACT Symposium on Theory of Computing},
  pages={690--703},
  year={2022}
}

@article{feldman2023simultaneous,
  title={Simultaneous 2nd price item auctions with no-underbidding},
  author={Feldman, Michal and Shabtai, Galia},
  journal={Games and Economic Behavior},
  volume={140},
  pages={316--340},
  year={2023},
  publisher={Elsevier}
}

@inproceedings{liu2024optimal,
  title={Optimal auction design with user coupons in advertising systems},
  author={Liu, Xiaodong and Fan, Zhikang and Ding, Yiming and Guo, Yuan and Zhang, Lihua and Li, Changcheng and Kong, Dongying and Li, Han and Shen, Weiran},
  booktitle={Proceedings of the Thirty-Third International Joint Conference on Artificial Intelligence},
  pages={2904--2912},
  year={2024}
}

@inproceedings{giannakopoulos2024discrete,
    title={Discrete Single-Parameter Optimal Auction Design},
    author={Giannakopoulos, Yiannis and Hahn, Johannes},
    booktitle={International Symposium on Algorithmic Game Theory},
    pages={165--183},
    year={2024},
    organization={Springer}
}

\newpage
\appendix
\section{OMITTED PROOFS FROM SECTION \ref{sec:pair}}
\subsection{Proof of Theorem \ref{Thm:Pareto_regular}}\label{Proof-Thm:Pareto_regular}
\begin{proof}
    Fix a randomized auction with interim allocation profile $\{\hat{x}_i(q_i)\}_{i=1}^n$ in the quantile space.
    We first transform the allocation rule into a new one by only changing a single buyer's interim allocation.
    In a single process, we keep all buyer's interim allocation rule except $i$ and increase buyer $i$'s interim allocation.
    For any quantile profile $\mathbf{q} = (q_i, \mathbf{q}_{-i}) \in [0, 1]^n$, define function
    \begin{align*}
        h_i(q_i, \mathbf{q}_{-i}) = 1 - \prod_{j} q_j - \sum_{j \neq i} \int^1_{q_j} \hat{x}_j(r_j) \mathrm{d} r_j.
    \end{align*}
    The first negative term $\prod_{j} q_j$ is linear in $q_i$, and the second negative term $\sum_{j \neq i} \int^1_{q_j} \hat{x}_j(r_j) \mathrm{d} r_j$ is independent of $q_i$.
    Thus, for any fixed $\mathbf{q}_{-i}$, the function $h_i(q_i, \mathbf{q}_{-i})$ is linear in $q_i$.
    Now define
    \begin{align*}
        \underline{h}_i(q_i) = \min_{\mathbf{q}_{-i}} h_i(q_i,\mathbf{q}_{-i}).
    \end{align*}
    Since $\underline{h}_i(q_i)$ is the pointwise minimum of linear functions of $q_i$, it is concave in $q_i$.
    It is also decreasing and satisfies
    \begin{align*}
        \underline{h}_i(1) = 0, \qquad \underline{h}_i(0) = 1 - \sum_{j \neq i} \int_0^1 \hat{x}_j(r_j) \mathrm{d}r_j.
    \end{align*}
    The total interim allocation of buyer $i$ is $\int_0^1 \hat{x}_i(r_i) \mathrm{d}r_i$ lies in the interval $[\underline{h}_i(1), \underline{h}_i(0)]$.
    Hence there exists a threshold $s_i^* \in [0, 1]$ such that $\int_0^1 \hat{x}_i(r_i) \mathrm{d}r_i = \underline{h}_i(s_i^*)$.
    Using this $s_i^*$, we define a new interim allocation for buyer $i$:
    \begin{align*}
        \hat{x}_i^*(r_i) &= \begin{cases}
            - \underline{h}_i'(r_i), \quad & r_i \in [s_i^*, 1]; \\
            0, \quad & r_i \in [0, s_i^*),
        \end{cases}
    \end{align*}
    where $\underline{h}_i'$ is the (almost everywhere defined) derivative of the concave function $\underline{h}_i$.
    Because $\underline{h}_i$ is concave, $\underline{h}_i'$ is nonincreasing, so $-\underline{h}_i'$ is nondecreasing.
    Since $\underline{h}_i$ is a concave envelope of a function with slopes in $[-1, 0]$, its derivative (where it exists) must also lie in $[-1, 0]$, and hence $- \underline{h}'_i \in [0, 1]$.
    Therefore $\hat{x}_i^*(r_i)$ is monotone in $r_i$.
    Moreover,
    \begin{align*}
        \int_0^1 \hat{x}_i^*(r_i) \mathrm{d}r_i &= \int_{s_i^*}^1 -\underline{h}_i'(r_i) \mathrm{d}r_i = \underline{h}_i(s_i^*) - \underline{h}_i(1) = \underline{h}_i(s_i^*) \\
        &= \int_0^1 \hat{x}_i(r_i)\mathrm{d}r_i,
    \end{align*}
    so buyer $i$'s total interim allocation is preserved.
    By the definition of $\underline{h}_i$, for any threshold profile $(s_i, \mathbf{s}_{-i})$, we have
    \begin{align*}
        \int^1_{s_i} \hat{x}_i^*(r_i) \mathrm{d} r_i + \sum_{j \neq i} \int^1_{s_j} \hat{x}_j(r_j) \mathrm{d} r_j \leq 1 - \prod_{j}s_j,
    \end{align*}
    so the Border feasibility constraints continue to hold when we replace $\hat{x}_i$ by $\hat{x}_i^*$ and leave all other buyers' interim allocations unchanged.

    Next, we compare welfare and revenue from buyer $i$.
    Let $v_i(q_i)$ denote the value of buyer $i$ at quantile $q_i$.
    The expected welfare contribution of buyer $i$ can be written as
    \begin{align*}
            \int^1_{q_i = 0} v_i(q_i) \hat{x}_i(q_i) \mathrm{d}q_i = \int^1_{q_i = 0} \left(\int_{t = q_i}^1 \hat{x}_i(t) \mathrm{d}t \right) \mathrm{d}v_i(q_i).
    \end{align*}
    Because $\hat{x}_i^*$ is obtained from a transformation that shifts buyer $i$'s allocation probability toward higher quantiles (and thus higher values) while preserving the total allocation mass, and because $v_i(q_i)$ is increasing in $q_i$, we have
    \begin{align*}
        &\int^1_{q_i = 0} \left(\int_{t = q_i}^1 \hat{x}_i(t) \mathrm{d}t \right) \mathrm{d}v_i(q_i) \\
        \leq &\int^1_{q_i = 0} \left(\int_{t = q_i}^1 \hat{x}^*_i(t) \mathrm{d}t \right) \mathrm{d}v_i(q_i) \\
        = &\int^1_{q_i = 0} v_i(q_i) \hat{x}_i^*(q_i) \mathrm{d}q_i.
    \end{align*}
    Thus buyer $i$'s contribution to welfare weakly increases.

    For revenue, we use the Myersonian representation in terms of the virtual value $\phi_i$.
    Under regularity, $\phi_i$ is increasing in the value and hence increasing in the quantile.
    The same monotone reshaping that improves welfare also weakly increases the expected virtual surplus from buyer $i$, and therefore weakly increases revenue from buyer $i$.

    We apply this transformation sequentially for buyers $i = 1, 2, \cdots, n$: first obtain $\hat{x}_1^*$ and $s_1^*$ from the original profile, then obtain $\hat{x}_2^*$ and $s_2^*$ using $\hat{x}_1^*$ and $\hat{x}_j$ for $j \geq 2$, and so on, until we obtain $\hat{x}_n^*$ and $s_n^*$ from $(\hat{x}_1^*, \cdots, \hat{x}_{n-1}^*, \hat{x}_n)$.
    The resulting interim allocation profile $\hat{\mathbf{x}}^* = (\hat{x}_1^*, \cdots, \hat{x}_n^*)$ is Border-feasible, monotone for each buyer, and yields weakly higher welfare and revenue than the original randomized auction.
    This completes the first step of the proof.
    
    The second step shows that $\hat{\mathbf{x}}^*$ can be implemented by a deterministic DSIC mechanism.
    We have the following lemma.
    \begin{lemma}
        $\hat{\mathbf{x}}^*$ can be implemented in a deterministic DSIC mechanism.
    \end{lemma}
    \begin{proof}
        Fix a buyer $i$ and a quantile $q_i$.
        Under the new allocation rule $\hat{\mathbf{x}}^*$, Border's constraint is still tight for some profiles of the other buyers' quantiles $\mathbf{q}_{-i}$.
        First, we show that all such tight profiles for $(i, q_i)$ are coordinatewise ordered.
        Suppose not, we assume there exist two tight profiles $\mathbf{q}_{-i}^1$ and $\mathbf{q}_{-i}^2$ such that for some buyers $j_1$ and $j_2$, we have $q^1_{j_1} < q^2_{j_1} \quad \text{and} \quad q^1_{j_2} > q^2_{j_2}$.
        Construct a new profile that mixes these two coordinates, that is $q_{j_1} = \frac{q^1_{j_1} + q^2_{j_1}}{2}$ and $q_{j_2} = \frac{q^1_{j_2} + q^2_{j_2}}{2}$, keeping the others fixed.
        Because the Border constraint is tight at $(q_i, \mathbf{q}_{-i}^1)$, feasibility of $\hat{\mathbf{x}}^*$ forces the item to be allocated to buyer $j_1$ at this mixed profile.
        Because the Border constraint is also tight at $(q_i, \mathbf{q}_{-i}^2)$, it simultaneously forces the item to be allocated to buyer $j_2$.
        This contradicts single-item feasibility.
        Hence, for fixed $(i,q_i)$, any two tight profiles $(\mathbf{q}_{-i}^1, \mathbf{q}_{-i}^2)$ must be coordinatewise comparable: either $\mathbf{q}_{-i}^1 \leq \mathbf{q}_{-i}^2$ or $\mathbf{q}_{-i}^2 \leq \mathbf{q}_{-i}^1$.
        Now, for each pair of buyers $i \neq j$ and each $q_i \in [s_i^*, 1]$, define
        \begin{align*}
            \lambda_{i, j}(q_i) = \max\{ & q_j \mid (q_i, \mathbf{q}_{-i}) \text{ makes Border's constraint} \\
            & \text{tight.} \}, q_i \in [s_i^*, 1].
        \end{align*}
        By the ordering above, this maximum is well-defined, and the profile $(q_i, \lambda_{i, 1}(q_i), \lambda_{i, 2}(q_i), \cdots)$ is itself a tight quantile profile.
        Writing $\lambda_i(q_i) = (\lambda_{i, j}(q_i))_{j \neq i}$, these tight profiles induce a set $\mathcal{L}$ of relations $l = (q_1, \cdots, q_n)$ among buyers' types.
        Using the same “no-crossing” argument as before, one shows that the relations in $\mathcal{L}$ are themselves mutually ordered.
        For each relation $l = (q_1, \cdots, q_n) \in \mathcal{L}$, define its total score $c(l) = \sum_i q_i$.
        Then, for each buyer $i$ and each $q_i \in [s_i^*, 1]$, define the individual score
        \begin{align*}
            c_i(q_i) = \max\{c(l) \mid l = (q_1, \cdots, q_n) \in \mathcal{L}\}, q_i \in [s_i^*, 1].
        \end{align*}
        When $q_i \in [0, s_i^*)$, let $c_i(q_i) = 0$.
        Thus, each type $q_i$ is assigned a scalar priority $c_i(q_i)$.
        We now specify a deterministic hierarchical allocation rule in quantile space.
        Given any profile $\mathbf{q} = (q_1, \cdots, q_n)$, we compute scores $c_i(q_i)$ for all buyers.
        Next,
        \begin{itemize}
            \item If $\max_i c_i(q_i) = 0$, the seller keeps the item.
            \item If there is a unique non-zero highest score, the buyer with the highest score wins.
            \item If there is multiple non-zero highest score, the highest buyer with the smallest index wins.
        \end{itemize}
        This rule is deterministic by construction. 
        Moreover, since the scores $c_i(q_i)$ are derived from ordered tight profiles, they are non-decreasing in $q_i$, and the hierarchical selection is monotone in each buyer's own quantile.
        By the standard single-parameter characterization, this monotone deterministic allocation can be supported by appropriate critical-value payments, so it is DSIC.
        Finally, because the scores and relations are built exactly from the tight Border profiles of $\hat{\mathbf{x}}^*$, the interim allocation of this hierarchical mechanism coincides with $\hat{\mathbf{x}}^*$.
        Thus $\hat{\mathbf{x}}^*$ is implementable by a deterministic DSIC mechanism.
    \end{proof}
\end{proof}

\subsection{Proof of Lemma \ref{lem:solutions_of_o2}}
\begin{proof}
    We show that the solutions of optimization problem (\ref{Objective2}) actually take the forms depicted in Figure \ref{fig: Special cases of piecewise auctions}, which represents special cases of piecewise auctions.
    
    \begin{figure}[htbp]
        \centering
        \subfigure[Maximum solution]{
            \centering
            \includegraphics[scale=0.19]{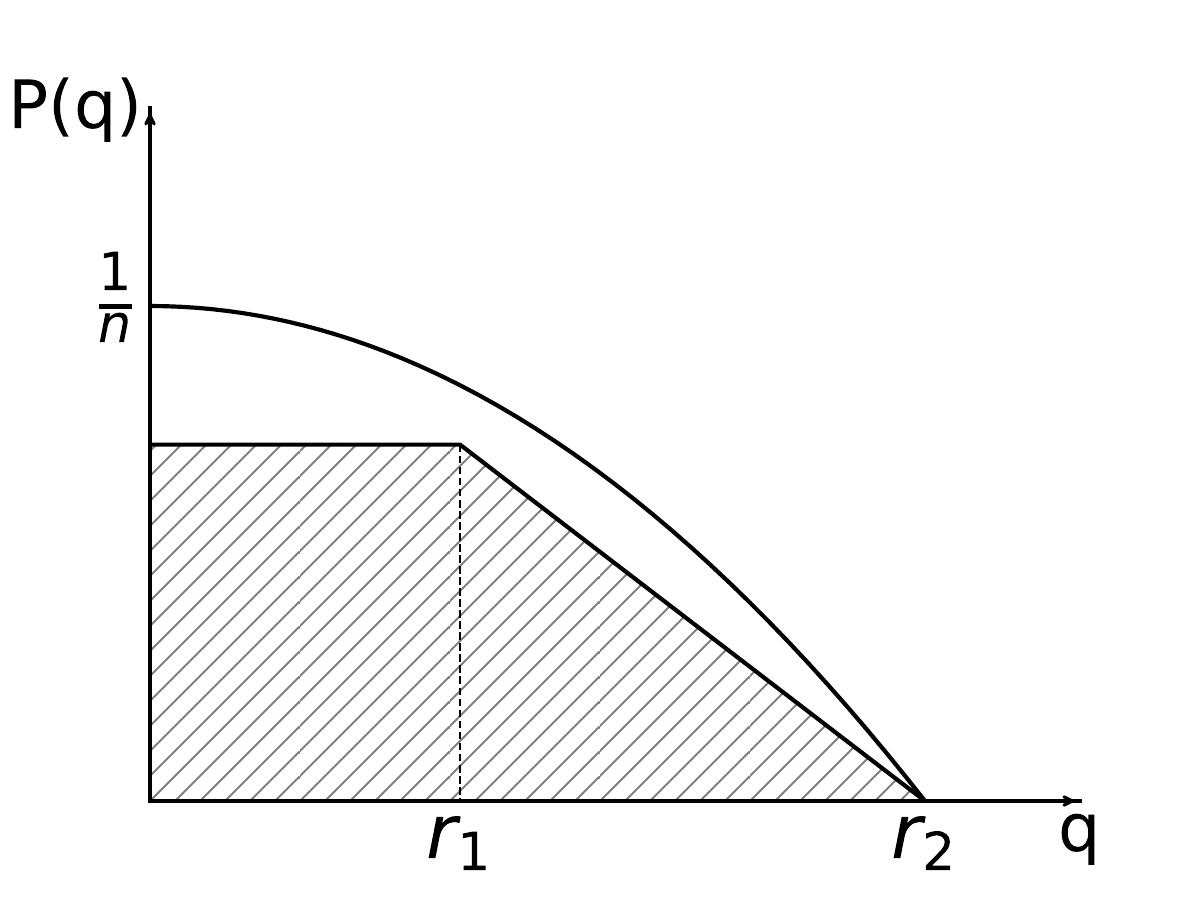}
        }
        \subfigure[Maximum solution]{
            \centering
            \includegraphics[scale=0.19]{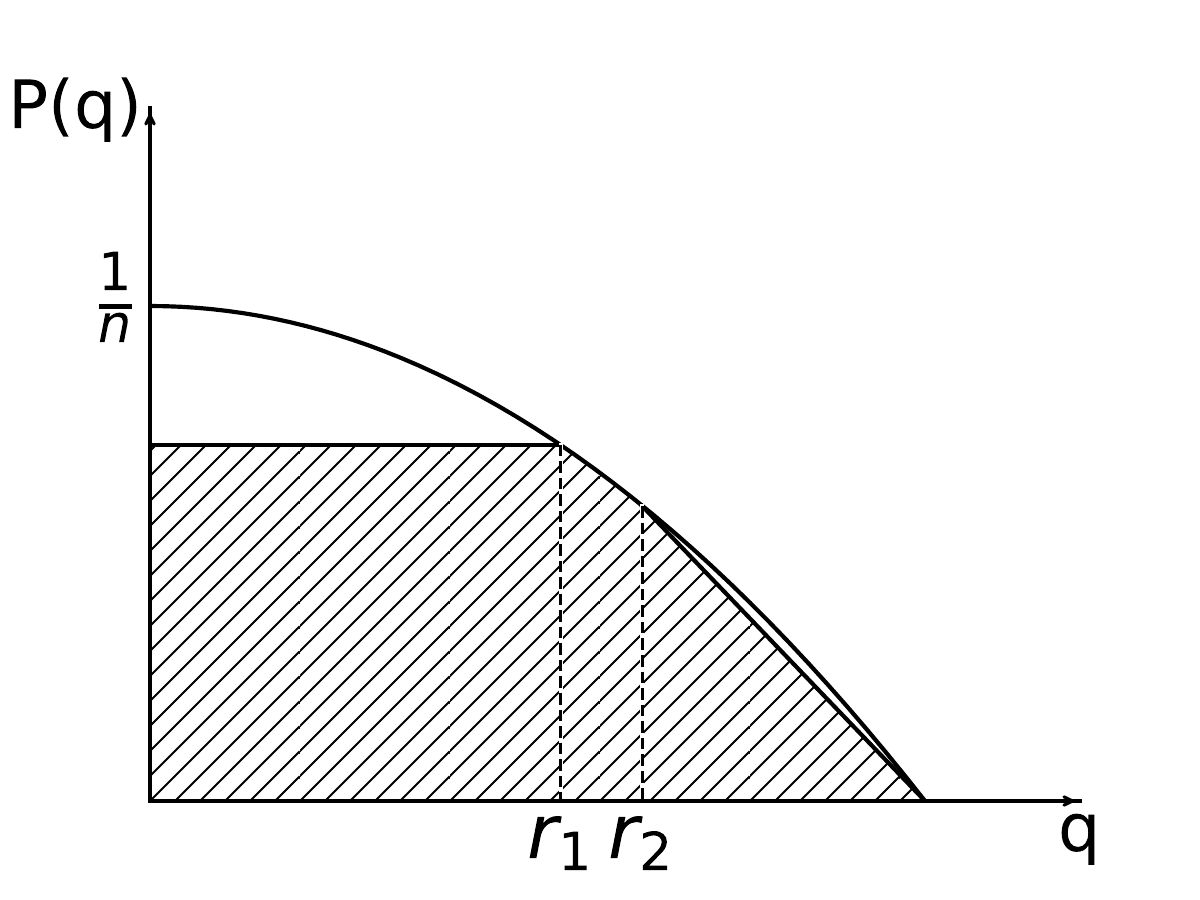}
        }
        \subfigure[Minimum solution]{
            \centering
            \includegraphics[scale=0.19]{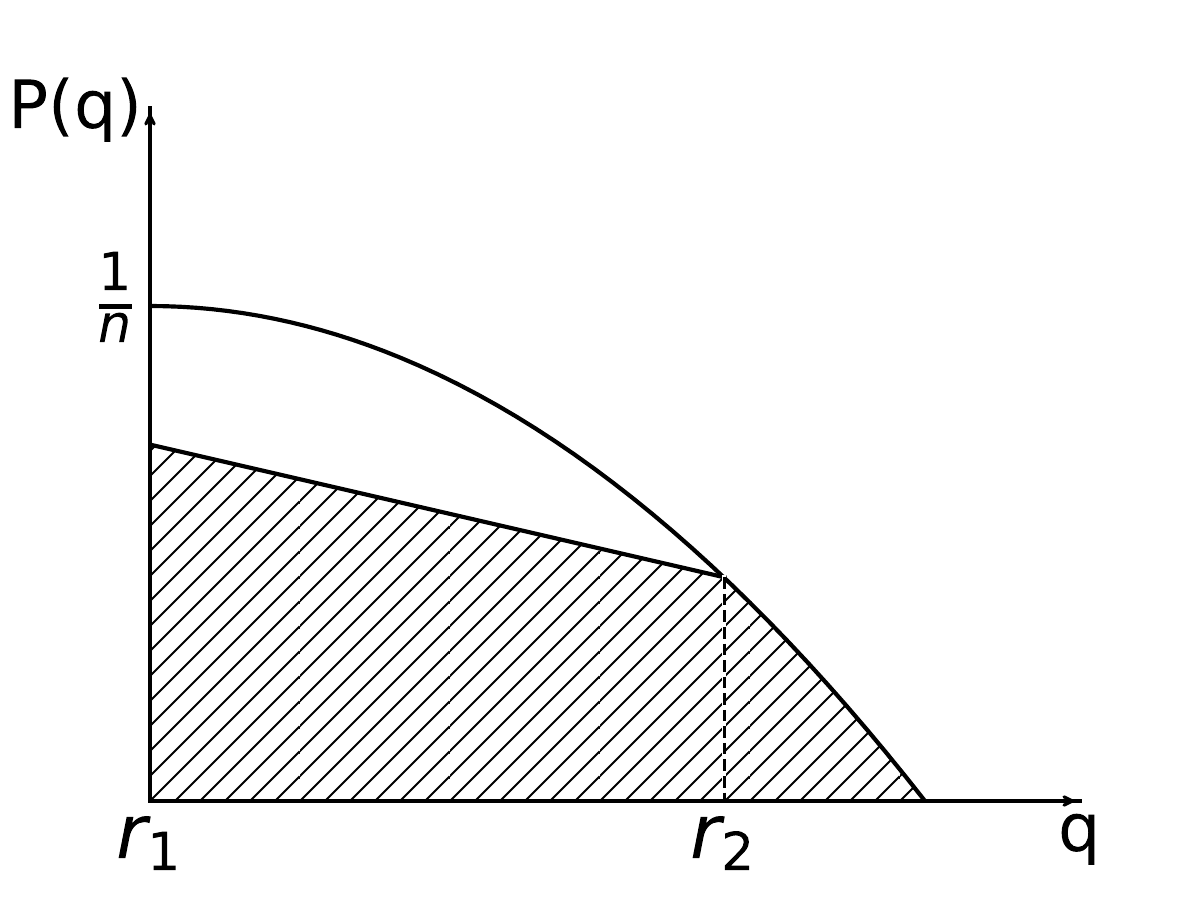}
        }
        \caption{Optimal Solutions of Optimization Problem (\ref{Objective2})}
        \label{fig: Special cases of piecewise auctions}
    \end{figure}
    
    First, we prove that in the maximum solution of optimization problem (\ref{Objective2}), the function $P(q)$ has the form in Figure \ref{fig: Special cases of piecewise auctions}(a) or \ref{fig: Special cases of piecewise auctions}(b), that is, $P(q)$ has the minimum stochastic order (for any feasible function $\hat{P}(q)$, $\int_t^1 P(q) \mathrm{d} q \leq \int_t^1 \hat{P}(q) \mathrm{d} q$, $\forall t \in [0,1]$).

    We fix the value of $P(0)$ and suppose that there exists another feasible solution, denoted by $\hat{P}(q)$, which is different from $P(q)$.
    By concavity, we know that $\hat{P}(q)$ crosses $P(q)$ at exactly one point $Q$, occurring at $q = q_0$, as shown in Figure \ref{fig: piecewise auctions}(a).
    Note that it is possible that the curve $P(q)$ and $\hat{P}(q)$ overlap.

    \begin{figure}[htbp]
        \centering
        \subfigure[Maximum solution]{
            \centering
            \includegraphics[scale=0.2]{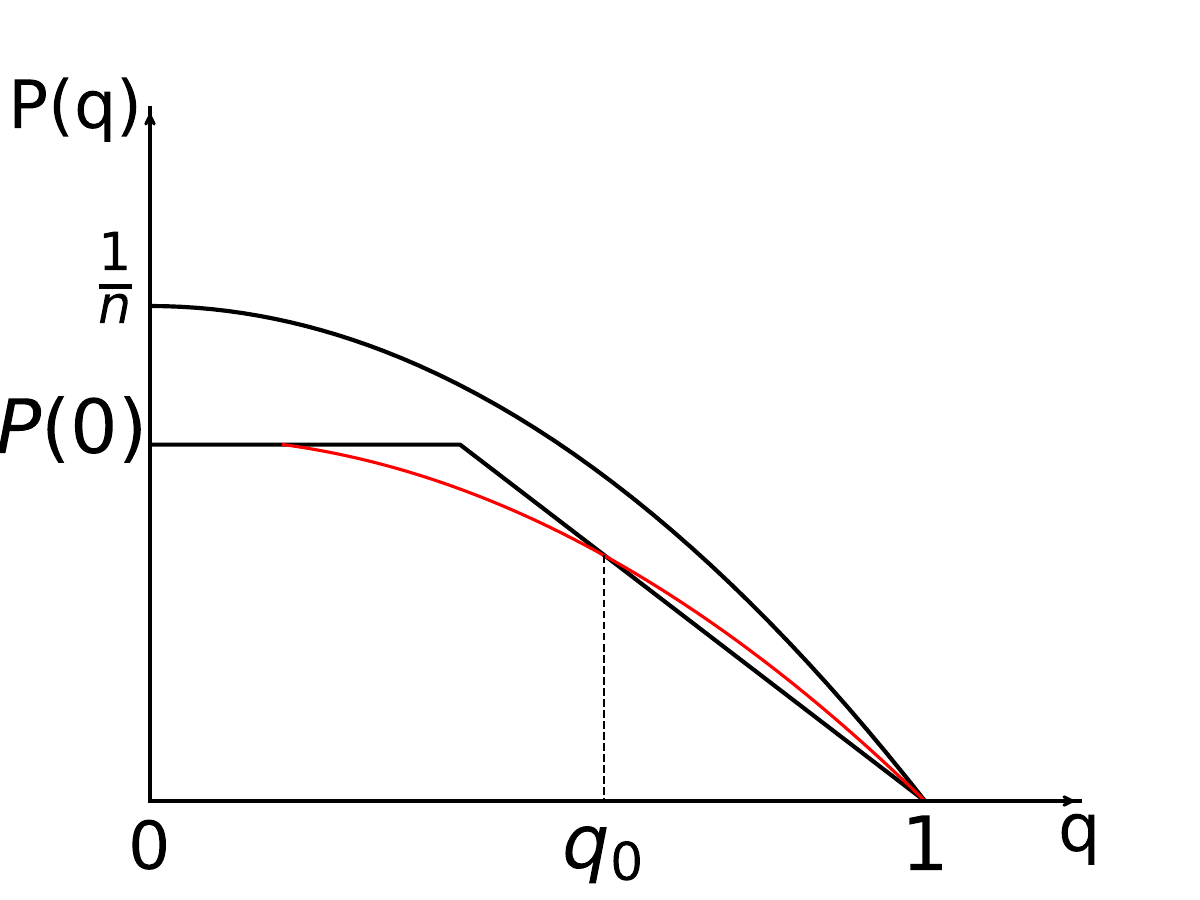}
        }
        \subfigure[Minimum solution]{
            \centering
            \includegraphics[scale=0.2]{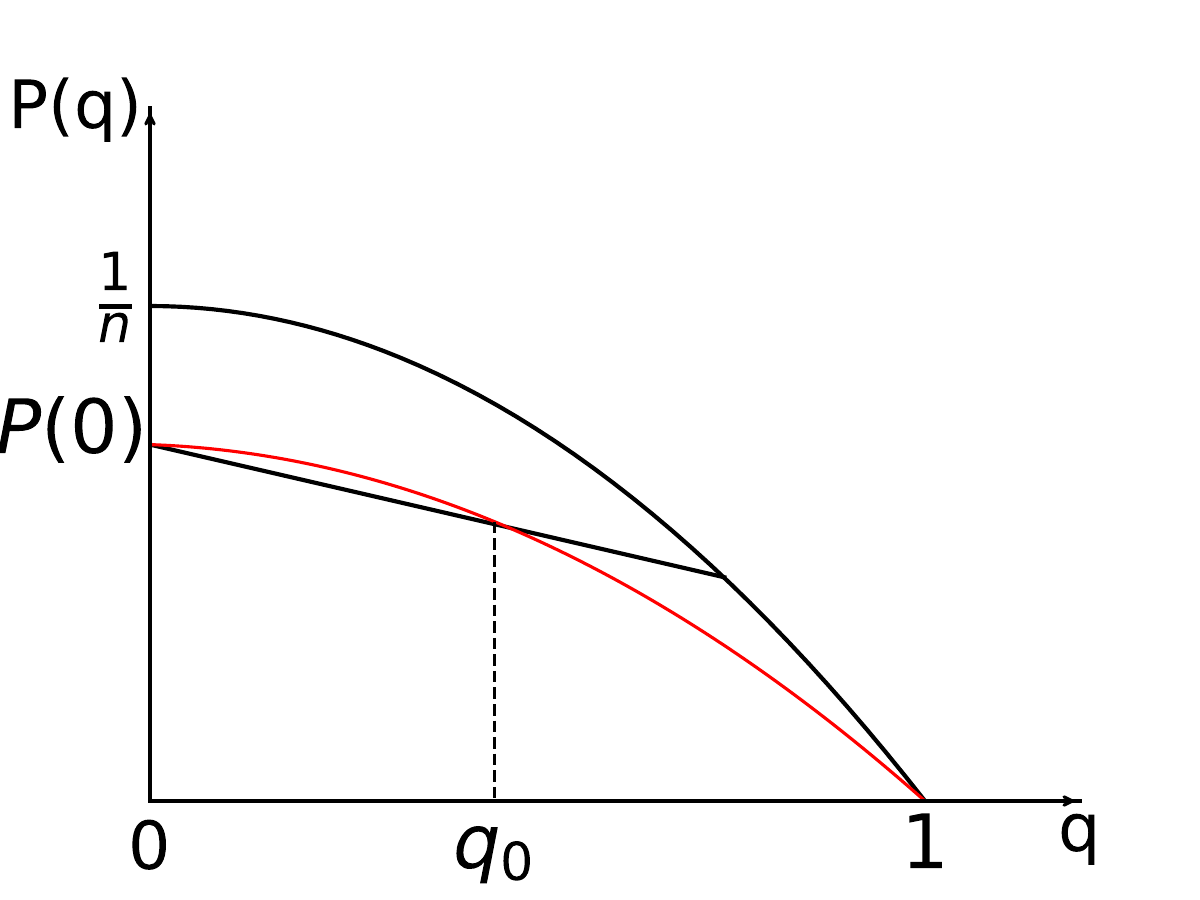}
        }
        \caption{Comparison Between $P(q)$ and $\hat{P}(q)$}
        \label{fig: piecewise auctions}
    \end{figure}
	
    Since $\int_0^1  v'(q)P(q) \mathrm{d} q = \int_0^1 v'(q) \hat{P}(q) \mathrm{d} q = \frac{c}{n}$, we have
    \begin{equation}\label{eq:lemma boundary}
        \int_0^{q_0} v'(q)(P(q) - \hat{P}(q)) \mathrm{d} q = \int_{q_0}^1 v'(q)(\hat{P}(q) - P(q)) \mathrm{d} q.
    \end{equation}

    We know that $P(q) \ge \hat{P}(q), \forall 0 \le q \le q_0$ and $P(q) \le \hat{P}(q), \forall q_0 \le q \le 1$.
    Together with $v''(q) \geq 0$, we have
    \begin{align}
        &\int_0^{q_0} v'(q)(P(q) - \hat{P}(q)) \mathrm{d} q \leq \int_{0}^{q_0} v'(q_0)(P(q) - \hat{P}(q)) \mathrm{d} q, \label{ineq1:lemma boundary} \\
        &\int_{q_0}^1 v'(q)(\hat{P}(q) - P(q)) \mathrm{d} q \geq \int_{q_0}^1 v'(q_0)(\hat{P}(q) - P(q)) \mathrm{d} q. \label{ineq2:lemma boundary}
    \end{align}

    Put equality \eqref{eq:lemma boundary}, and inequalities \eqref{ineq1:lemma boundary} and \eqref{ineq2:lemma boundary} together, we have
    \begin{align*}
        \int_0^{q_0} (P(q) - \hat{P}(q)) \mathrm{d} q \geq \int_{q_0}^1 (\hat{P}(q) - P(q)) \mathrm{d} q.
    \end{align*}

    Similarly, due to $\hat{\phi}''(q) \le 0$, the difference of revenue between the solution with $P(q)$ and $\hat{P}(q)$ is
    \begin{align*}
        & n \int_0^{q_0} \hat{\phi}'(q)(P(q) - \hat{P}(q)) \mathrm{d} q - n \int_{q_0}^1 \hat{\phi}'(q)(\hat{P}(q) - P(q)) \mathrm{d} q \\
        \geq ~ & n \hat{\phi}'(q_0) \cdot \left(\int_0^{q_0} (P(q) - \hat{P}(q)) \mathrm{d} q - \int_{q_0}^1 (\hat{P}(q) - P(q)) \mathrm{d} q \right) \\
        \geq ~ & 0.
    \end{align*}

    Finally, we prove that $\hat{\phi}'(q_0) \ge 0$.
    If $\hat{\phi}'(q_0) < 0$, then $\hat{\phi}'(q) < 0$ for all $q > q_0$ according to $\hat{\phi}''(q) \le 0$.
    Thus, $\hat{\phi}(q_0) > \hat{\phi}(1) = 1$, which leads to a contradiction.
    So, we have $\hat{\phi}'(q_0) \geq 0 $.
    Therefore, we solve the maximum solution.

    For the minimum solution, we prove that the function $P(q)$ has the form in Figure \ref{fig: Special cases of piecewise auctions}(c), that is $P(q)$ has the maximum stochastic order.
    The proof is similar and omitted.
\end{proof}

\subsection{Proof of Lemma \ref{Lem:sum_equal}}
\begin{proof}
    The revenue and the welfare for the auction with interim allocation profile $(x_1(v), \cdots, x_n(v))$ are
    \begin{align*}
        \textrm{WEL} & = \sum_{i = 1}^n \int_0^1 v x_i(v) f(v) \mathrm{d} v = \int_0^1 v f(v) \sum_{i = 1}^n x_i(v) \mathrm{d} v, \\
        \textrm{REV} & = \sum_{i = 1}^n \int_0^1 \phi(v) x_i(v) f(v) \mathrm{d} v = \int_0^1 \phi(v) f(v) \sum_{i = 1}^n x_i(v) \mathrm{d} v.
    \end{align*}
    Because welfare and revenue are only related to the interim allocation profile, $\sum_{i = 1}^n x_i(v) = \sum_{i = 1}^n y_i(v)$, thus two auctions' revenue and welfare are the same.
\end{proof}

\subsection{Proof of Lemma \ref{lem:implementation_of_simple_auctions}}
\begin{proof}
    For piecewise auctions satisfying the first conditions in Definition \ref{def:piecewise auctions}, the corresponding interim allocations are as follows:
    \begin{itemize}
        \item $\hat{x}(q) = 0$, for $0 < q \leq r_1$,
        \item $\hat{x}(q) = k$, for $r_1 < q \leq r_2$,
        \item $\hat{x}(q) = q^{n-1}$, for $r_2 < q \leq 1$.
    \end{itemize}
    Define $A = \{i | q_i > r_2\}$, $B = \{i | r_1 \leq q_i \leq r_2\}$ and construct the deterministic allocation rule as follows: 
    \begin{itemize}
        \item If $|A| \neq 0$, set $\hat{x}_j(q_1, q_2, \cdots, q_n) = 1$, where $j$ is the minimum element in $A$ such that $q_j = \max_{i \in A} \{ q_i \}$, and set $\hat{x}_i(q_1, q_2, \cdots, q_n) = 0 , \forall i \neq j$.
        \item If $|A| = 0$, $|B| > 0$ and $\frac{r_2^{n-1}} n \leq k \leq \frac{r_2^n - r_1^n}{n(r_2 - r_1)}$:
            \begin{itemize}
                \item If $1 \in B$, then set $\hat{x}_1(q_1, q_2, \cdots, q_n) = 1$, set $\hat{x}_i(q_1, q_2, \cdots, q_n) = 0, \forall i > 1$.
                \item If $1 \notin B$, suppose that $j$ is the minimum element in $B$,
                    \begin{itemize}
                        \item If $q_1 \leq \lambda r_1$, set $\hat{x}_j(q_1, q_2, \cdots, q_n) = 1$, $\hat{x}_i(q_1, q_2, \cdots, q_n) = 0$, $\forall i \neq j$.
                        \item If $\lambda r_1 < q_1 < r_1$, set $\hat{x}_i(q_1, q_2, \cdots, q_n) = 0, \forall i$. ($\lambda$ ranges from 0 to $1$.)
                    \end{itemize}
            \end{itemize}
        \item If $|A| = 0$, $|B| > 0$ and $ k < \frac{r_2^{n-1}}{n}$:
            \begin{itemize}
                \item If $1 \in B$, 
                    \begin{itemize}
                        \item If $q_2 \leq \lambda r_2$, set $\hat{x}_1(q_1, q_2, \cdots, q_n) = 1, \hat{x}_i(q_1, q_2, \cdots, q_n) = 0, \forall i > 1$.
                        \item If $\lambda r_2 < q_2 \leq r_2$, set $\hat{x}_i(q_1, q_2, \cdots, q_n) = 0, \forall i$. ($\lambda$ ranges from 0 to $1$.)
                    \end{itemize}
                \item If $1 \notin B$, set $\hat{x}_i(q_1, q_2, \cdots, q_n) = 0, \forall i$.
            \end{itemize}
        \item If $|A| = 0$ and $|B| = 0$, set $\hat{x}_i(q_1, q_2, \cdots, q_n) = 0, \forall i$.
    \end{itemize}
    The payment rule is according to the payment identity. 

    It is obvious that $\hat{x}_i(q) = 0, q < r_1$ and $\hat{x}_i(q) = q^{n-1}, q > r_2$.
    We then analyze the case $r_1 \leq q \leq r_2$.

    If $\frac{r_2^{n-1}}{n} \leq k \leq \frac{r_2^n - r_1^n}{n(r_2 - r_1)}$, we have $\hat{x}_1(q) = Pr[q_j \leq r_2, \forall j \neq 1] = r_2^{n-1}$.
    And $\forall i > 1$,
    \begin{align*}
        \hat{x}_i(q) = & Pr[q_j < r_1, \forall j < i] \cdot Pr[q_j < r_2, \forall j > i] \\
        & \cdot Pr[q_1 < \lambda r_1 | q_1 < r_1] \\
        = & \lambda r_1^{i-1} r_2^{n-i}.
    \end{align*}
    Thus, we have
    \begin{align*}
        \sum_{i = 1}^n \hat{x}_i(q) = r_2^{n - 1} + \sum_{i = 2}^n \lambda r_1^{i - 1} r_2^{n - i} \in [r_2^{n-1}, \frac{r_2^n - r_1^n}{r_2 - r_1}].
    \end{align*}
    Therefore, there exists $\lambda \in [0, 1]$ such that $\sum_{i = 1}^n \hat{x}_i(q) = nk$.

    If $ k < r_2^{n-1}$, $\hat{x}_i(q) = 0, \forall i > 1$.
    And we have
    \begin{align*}
        \hat{x}_1(q) &= Pr[q_j \leq r_2, \forall j \neq i] \cdot Pr[q_2 \leq \lambda r_2 | q_2 \leq r_2] \\
        &= \lambda r_2^{n-1} \in [0, r_2^{n-1}].
    \end{align*}
    Thus, there exists $\lambda \in [0, 1]$ such that $\sum_{i = 1}^n \hat{x}_i(q) = nk$.

    For piecewise auctions satisfying the second conditions in Definition \ref{def:piecewise auctions}, the corresponding interim allocations are as follows:
    \begin{itemize}
        \item $\hat{x}(q) = 0$, for $0 < q \leq r_1$,
        \item $\hat{x}(q) = q^{n-1}$, for $r_1 < q \leq r_2$,
        \item $\hat{x}(q) = \frac{1 - r_2^n}{n(1 - r_2)}$, for $r_2 < q \leq 1$.
    \end{itemize}

    Define $A = \{i | q_i > r_2\}$, $B = \{i | r_1 \leq q_i \leq r_2\}$ and construct the deterministic allocation rule as follows:
    \begin{itemize}
        \item If $|A| > 0$, set $\hat{x}_j(q_1, q_2, \cdots, q_n) = 1$ where $j$ is the minimum element in $A$, and set $\hat{x}_i(q_1, q_2, \cdots, q_n) = 0, \forall i \neq j$. 
        \item If $|A| = 0$ and $|B| > 0$, set $\hat{x}_j(q_1, q_2, \cdots, q_n) = 1$ where $j$ is the minimum element in $B$ such that $q_j = \max_{i \in B}\{q_i\}$, and set $\hat{x}_i(q_1, q_2, \cdots, q_n) = 0, \forall i \neq j$. 
        \item If $|A| = 0$ and $|B| = 0$, set $\hat{x}_i(q_1, q_2, \cdots, q_n) = 0, \forall i$.
    \end{itemize}
    The payment rule is according to the payment identity.

    For all $q > r_2$, we have $\hat{x}_i(q) = Pr[q_j < r_2, \forall j < i] = r_2^{i - 1}$, thus $\sum_{i = 1}^n \hat{x}_i(q) = \sum_{i = 1}^n r_2^{i - 1} = \frac{1 - r_2^n}{1 - r_2}$.
    
    For all $r_1 \leq q \leq r_2$, $\hat{x}_i(q) = Pr[q_j < q, \forall j \neq i] = q^{n - 1}$ , thus $\sum_{i = 1}^n \hat{x}_i(q) = nq^{n-1}$.
    
    For all $q < r_1$,  $\hat{x}_i(q) = 0$.
    
    So we have $\sum_{i = 1}^n \hat{x}_i(q) = n\hat{x}(q)$.
    By Lemma \ref{Lem:sum_equal}, the (revenue, welfare) pairs of piecewise auctions is implementable by deterministic DSIC auctions.
\end{proof}

\subsection{Proof of Lemma \ref{lem:inner-pairs}}
\begin{proof}
    Fix a welfare level $c$.
    Let $\mathcal{M}^a$ and $\mathcal{M}^b$ be two piecewise auctions (Definition \ref{def:piecewise auctions}) with $\mathrm{WEL}(\mathcal{M}^a) = \mathrm{WEL}(\mathcal{M}^b) = c$.
    We construct a continuous transfer (via a sequence of piecewise auctions, all having welfare $c$) from each auction to a common canonical piecewise auction.
    Throughout, write the interim allocation as $\hat{x}(q)$, $q \in [0, 1]$, so
    \begin{align*}
        \mathrm{WEL}(\hat{x}) = n \int_0^1 v(q) \hat{x}(q) \mathrm{d}q,
    \end{align*}
    where $v(q) = F^{-1}(q)$.
    
    Define the ``reserve-type'' interim allocation
    \begin{align*}
        \hat{x}^{\rho}(q) := \mathbb{I}\{q > \rho\} q^{n-1}.
    \end{align*}
    Its welfare is
    \begin{align*}
        W(\rho) := n \int_{\rho}^{1} v(q)q^{n-1} \mathrm{d}q,
    \end{align*}
    which is continuous and (weakly) decreasing in $\rho$.
    Hence there exists $\rho^\star \in [0, 1]$ such that $W(\rho^\star) = c$.
    Let $\mathcal{M}^\star$ denote the piecewise auction implementing $\hat{x}^{\rho^\star}$.
    Note $\hat{x}^{\rho^\star}$ is exactly a Condition 2 piecewise auction in Definition \ref{def:piecewise auctions} with parameters $(r_1, r_2) = (\rho^\star, 1)$.
    It suffices to show: any piecewise auction $\mathcal{M}$ with welfare $c$ can be continuously transferred to $\mathcal{M}^\star$ through piecewise auctions of welfare $c$.
    
    First, we discuss Condition 2 in Definition \ref{def:piecewise auctions}.
    A Condition 2 piecewise auction is characterized by parameters $0 \le r_1 < r_2 \le 1$ and interim allocation
    \begin{align*}
        \hat{x}(q) = \begin{cases}
            0, & q \le r_1; \\
            q^{n-1}, & r_1 < q \le r_2; \\
            \kappa(r_2) := \dfrac{1 - r_2^n}{n(1 - r_2)}, & r_2 < q \le 1.
        \end{cases}
    \end{align*}
    Define $r_2(t) := (1 - t) r_2 + t \text{ for } t \in [0, 1]$.
    For each $t$, choose $r_1(t) \in [0, r_2(t))$ so that the resulting Condition 2 piecewise auction $\mathcal{M}_t$ satisfies
    $\mathrm{WEL}(\mathcal{M}_t) = c$.
    This is well-defined and continuous in $t$ because, for fixed $r_2(t)$, welfare is a continuous and strictly decreasing function of $r_1$ (increasing $r_1$ removes positive allocation $q^{n-1}$ on $(r_1, r_2(t)])$.
    Thus $t \mapsto r_1(t)$ is continuous, $\mathcal{M}_t$ stays within Condition 2, and $\mathcal{M}_1$ has $(r_1(1), r_2(1)) = (\rho^\star, 1)$, i.e. $\mathcal{M}_1 = \mathcal{M}^\star$.
    
    Second, we discuss Condition 1 in Definition \ref{def:piecewise auctions}.
    A Condition 1 piecewise auction has parameters $0 \le r_1 < r_2 \le 1$ and constant $k$ with
    \begin{align*}
        0 \le k \le k_{\max}(r_1, r_2) := \frac{r_2^n - r_1^n}{n (r_2 - r_1)},
    \end{align*}
    and interim allocation
    \begin{align*}
        \hat{x}(q) = \begin{cases}
            0, & q \le r_1, \\
            k, & r_1 < q \le r_2, \\
            q^{n-1}, & r_2 < q \le 1.
        \end{cases}
    \end{align*}
    Let $\mathcal{M}$ be such an auction with welfare $c$.
    Define
    \begin{align*}
        T(r) := n \int_{r}^{1} v(q) q^{n-1} \mathrm{d}q, \quad (\text{tail welfare}).
    \end{align*}
    Since $\mathrm{WEL}(\mathcal{M}) = T(r_2) + n k \int_{r_1}^{r_2} v(q) \mathrm{d}q = c$ with $k \ge 0$, we have $T(r_2) \le c = T(\rho^\star)$, hence $r_2 \ge \rho^\star$.
    Now set $r_2(t) := (1 - t)r_2 + t \rho^\star$, so $r_2(t) \in [\rho^\star, r_2]$.
    Keep $r_1$ fixed.
    For each $t$, define $k(t)$ as the unique value making welfare equal to $c$, we have
    \begin{align*}
        k(t) := \frac{c - T(r_2(t))}{n \int_{r_1}^{r_2(t)} v(q) \mathrm{d}q}.
    \end{align*}
    This is well-defined because $T(r_2(t)) \le T(\rho^\star) = c$ and $\int_{r_1}^{r_2(t)} v(q) \mathrm{d}q > 0$.
    Moreover, feasibility $k(t) \le k_{\max}(r_1, r_2(t))$ holds because, for fixed $(r_1, r_2(t))$, welfare is affine increasing in $k$ on the feasible interval $[0, k_{\max}]$.
    Since $k(0) = k$ is feasible and $t \mapsto (r_2(t), k(t))$ is continuous, the equality $\mathrm{WEL} = c$ selects a feasible $k(t) \in [0, k_{\max}(r_1, r_2(t))]$ for all $t$.
    Thus each $\mathcal{M}_t$ defined by $(r_1, r_2(t), k(t))$ is a Condition 1 piecewise auction with welfare $c$, and $t \mapsto \mathcal{M}_t$ is continuous.
    At $t = 1$, $r_2(1) = \rho^\star$ and $T(\rho^\star) = c$, hence $k(1) = 0$.
    Therefore
    \begin{align*}
        \hat{x}_{M_1}(q) = \mathbb{I}\{q > \rho^\star\} q^{n-1} = \hat{x}^{\rho^\star},
    \end{align*}
    so $\mathcal{M}_1 = \mathcal{M}^\star$.
    
    Apply the above transfer process, we can transfer $\mathcal{M}^a$ continuously to $\mathcal{M}^\star$ through piecewise auctions with welfare $c$, and similarly transfer $\mathcal{M}^b$ to $\mathcal{M}^\star$.
    Concatenating the first transfer with the reverse of the second yields a sequence of piecewise auctions that continuously transfers $\mathcal{M}^a$ to $\mathcal{M}^b$ while keeping welfare fixed.
\end{proof}

\subsection{The Proof of Theorem \ref{thm:equal_to_random}}\label{Proof:equal_to_random}
\begin{proof}
    Fix a welfare level $c$.
    Consider the optimization problem \ref{Objective2} that maximizes (resp. minimizes) revenue over all feasible symmetric randomized auctions subject to $\mathrm{WEL} = c$.
    Let $R_{\max}(c)$ and $R_{\min}(c)$ denote the optimal values.
    
    By Lemma \ref{lem:solutions_of_o2}, there exist optimal solutions to these two problems whose induced auctions are piecewise.
    Hence there are two piecewise auctions $\mathcal{M}^{\max}$ and $\mathcal{M}^{\min}$ such that
    \begin{align*}
        \mathrm{WEL}(\mathcal{M}^{\max}) &= \mathrm{WEL}(\mathcal{M}^{\min}) = c, \\
        \mathrm{REV}(\mathcal{M}^{\max}) &= R_{\max}(c), \quad \mathrm{REV}(\mathcal{M}^{\min}) = R_{\min}(c).
    \end{align*}
    
    By Lemma \ref{lem:implementation_of_simple_auctions}, each piecewise auction can be implemented by a deterministic DSIC mechanism with the same total interim allocation.
    By Lemma \ref{Lem:sum_equal}, matching total interim allocation implies matching expected welfare and expected revenue.
    Therefore, the two boundary points $(R_{\max}(c), c)$ and $(R_{\min}(c), c)$ are both achievable by deterministic DSIC auctions.
    Next, Lemma \ref{lem:inner-pairs} implies that $\mathcal{M}^{\max}$ can be continuously transferred to $\mathcal{M}^{\min}$ through a sequence of piecewise auctions while keeping welfare fixed at $c$.
    Along this transfer, revenue varies continuously, so every intermediate revenue level $r \in [R_{\min}(c), R_{\max}(c)]$ is achieved by some piecewise auction $\mathcal{M}$ with $\mathrm{WEL}(\mathcal{M}) = c$ and $\mathrm{REV}(\mathcal{M}) = r$.
    Applying Lemmas \ref{lem:implementation_of_simple_auctions} and \ref{Lem:sum_equal} again, each such point $(r, c)$ is implementable by a deterministic DSIC auction.
    Finally, any randomized auction with welfare $c$ must have revenue in $[R_{\min}(c), R_{\max}(c)]$ by definition of $R_{\min}(c)$ and $R_{\max}(c)$.
    Hence every $(\mathrm{REV}, \mathrm{WEL})$ pair attainable by randomized auctions at welfare level $c$ is attainable by deterministic DSIC auctions.
    Since $c$ was arbitrary, the entire $(\mathrm{REV}, \mathrm{WEL})$ region attainable by randomized auctions is attainable by deterministic DSIC auctions, proving Theorem \ref{thm:equal_to_random}.
\end{proof}

\section{OMITTED PROOFS FROM SECTION \ref{sec:interim}}
\subsection{Proof of Theorem \ref{Thm:Border_deterministic}}\label{Proof-Border_deterministic}
\begin{proof}
    Consider the plane
    \begin{align*}
        A = \{(q_1, q_2), 0 \leq q_1 \leq 1, 0 \leq q_2 \leq 1\}.
    \end{align*}
    We first prove the ``only if'' direction: for any deterministic DSIC auction with interim allocation profiles $(\hat{x}_1(q), \hat{x}_2(q))$, condition \eqref{Border_deterministic} holds.
    We construct a corresponding three-coloring $C$ of plane $A$: $C(q_1, q_2) \in \{0, 1, 2\}$, which means that the item is allocated to no one, buyer 1 or buyer 2.
    
    By \cite{nisan2007algorithmic}, we know that if a DSIC auction is deterministic, then for any $i$, there exists a value $c_i(\mathbf{v}_{-i})$ that is independent of $v_i$, called the critical value, such that
    \begin{itemize}
        \item If $v_i \geq c_i(\mathbf{v}_{-i})$, then $x_i(v_i, \mathbf{v}_{-i}) = 1$, $p_i(v_i, \mathbf{v}_{-i}) = c_i(\mathbf{v}_{-i})$.
        \item If $v_i < c_i(\mathbf{v}_{-i})$, then $x_i(v_i, \mathbf{v}_{-i}) = 0$, $p_i(v_i, \mathbf{v}_{-i}) = 0$.
    \end{itemize}
    
    Thus, for each quantile $q_2$ of buyer 2, there is a function $c_1(q_2)$ which only depends on $q_2$, called the critical value for buyer 1, i.e, $\hat{x}_1(q, q_2) = 1$ if and only if $q \geq c_1(q_2)$.
    
    Next, we rearrange the coloring $C$ to $C'$ with fixed interim allocations (see Figure \ref{fig: deter_border_arrang}).
    \begin{itemize}
        \item $C'(q_1, q_2) = 1$ for all $q_2 \leq \hat{x}_1(q_1)$,
        \item $C'(q_1, q_2) = 0$ for all $q_2 > \hat{x}_1(q_1)$ and $C(q_1, q_2) = 1$.
        \item $C'(q_1, q_2) = 2$ for all $q_2 > \hat{x}_1(q_1)$ and $C(q_1, q_2) = 2$.
    \end{itemize}
    In fact, the rearrangement moves down all the area of buyer 1.
    Next, we show that the interim allocations keep unchanged.
    For buyer 1 with any quantile $q_1$, the new interim allocation is $\int_{q_2} \mathbf{I}[C'(q_1, q_2) = 1] \mathrm{d} q_2 = \hat{x}_1(q_1)$, where $\mathbf{I}$ is the indicator function.
    So the interim allocations remain the same for buyer 1.
    
    For buyer 2, we prove that all areas that belong to buyer 2 remain the same.
    Otherwise, suppose that there is a point $(q_1, q_2)$ such that $C(q_1, q_2) = 2$ and $C'(q_1, q_2) \neq 2$.
    Then by the construction, it is easy to see that $q_2 < \hat{x}_1(q_1)$.
    By DSIC, we have $C(q_1, q) = 2$ for all $q > q_2$.
    Therefore, we have $\hat{x}_1(q_1) \leq 1 - \int_{q_2}^1 \mathbf{I}[C(q_1, q) = 2] \mathrm{d} q_2 = q_2 < \hat{x}_1(q_1)$, contradiction.
    Here we ignore the case where $q_2 = \hat{x}_1(q_1)$ is the minimum value of $q_2$ such that $C(q_1, q_2) = 2$, of which the probability is measured zero.
    
    \begin{figure}[htbp]
        \centering
        \includegraphics[scale = 0.18]{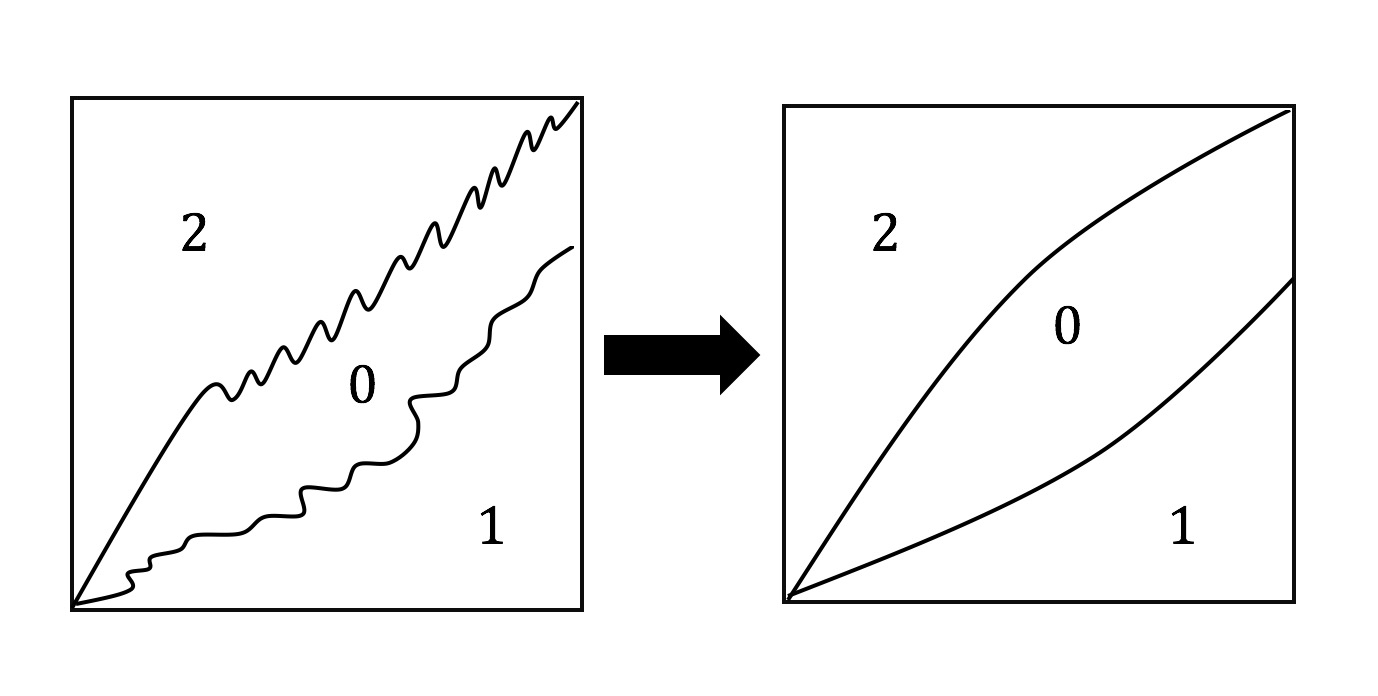}
        \caption{Rearrangment on three-coloring of deterministic auctions}
        \label{fig: deter_border_arrang}
    \end{figure}
    
    Do the same rearrangement for buyer 2, as shown in Figure \ref{fig: deter_border_arrang}, the area belongs to buyer 1 and buyer 2 are partitioned by two separate monotone curves, $q_2 = \hat{x}_1(q_1)$ and $q_1 = \hat{x}_2(q_2)$.
    Since the two curves does not cross each other, it is possible that the two curves overlap.
    The results still hold as the overlapped area is measured zero, we have $\forall q, ~ \hat{x}_2(\hat{x}_1(q)) \leq q$.
    
    If the item is always sold, then no point in plane $A$ is colored 0.
    So the two curves $q_2 = \hat{x}_1(q_1)$ and $q_1 = \hat{x}_2(q_2)$ totally overlap (the colored 0 area does not exists).
    It follows that $\forall q, ~\hat{x}_2(\hat{x}_1(q)) = q$.
    
    For the ``if'' direction, given (\ref{Border_deterministic}), we can simply construct the DSIC auction by using the two curves, $q_2 = \hat{x}_1(q_1)$ and $q_1 = \hat{x}_2(q_2)$, to partition the plane $A$.
    The condition (\ref{Border_deterministic}) guarantees that the two curves do not cross each other, thus the construction is feasible.
    Moreover, if (\ref{Border_deterministic2}) holds, then the two curves fully overlap, all points are colored either 1 or 2.
\end{proof}

\subsection{Proof of Corollary \ref{lem:corollary}}
\begin{proof}
    Consider the 2-buyer case with i.i.d. continuous priors and the following randomized BIC auction:
    \begin{itemize}
        \item with probability $p$, run VCG auction.
        \item with probability $1 - p$, send the item for free to buyer 1.
    \end{itemize}
    We analyze it in the quantile space.
    Then, for any quantile $q = F(v)$, the sum of the interim allocation is $2pq + 1 - p$.
    We assume that there exists a deterministic DSIC auction with interim allocation profile $(\hat{x}_1(q), \hat{x}_2(q))$ such that
    \begin{align*}
        \hat{x}_1(q) + \hat{x}_2(q) = 2pq + 1 - p.
    \end{align*}
    Since
    \begin{align*}
        \int_0^1 (\hat{x}_1(q) + \hat{x}_2(q)) \mathrm{d} q = \int_0^1 (2pq + 1 - p) \mathrm{d} q = 1,
    \end{align*}
    so the item is allocated with probability 1.
    By Theorem ~\ref{Thm:Border_deterministic}, we have $\hat{x}_2(\hat{x}_1(q)) = q$ for all $q$.
    Since $\hat{x}_1(0) + \hat{x}_2(0) = 1 - p$, either $\hat{x}_1(0)$ or $\hat{x}_2(0)$ must be equal to 0.
    Otherwise, if $\hat{x}_1(0) > 0$ and $\hat{x}_2(0) > 0$, then $\hat{x}_2(\hat{x}_1(0)) \geq \hat{x}_2(0) > 0$, contradiction.
    Thus, without loss of generality, assume that $\hat{x}_1(0) = 1 - p$ and $\hat{x}_2(0) = 0$.
    
    Therefore, we can get a sequence:
    \begin{align*}
        u_{n} + u_{n + 2} = 2pu_{n + 1} + 1 - p, u_1 = 0, u_2 = 1 - p,
    \end{align*}
    and $\hat{x}_1(u_n) = u_{n + 1}$, $\hat{x}_2(u_n) = u_{n - 1}$.
    
    So if the sequence $\{u_n\}$ never reaches 1 exactly, then the interim allocation is not feasible.
    It is not hard to check that when $p = 0.75$, $u_{5} = 1.0312 > 1$.
    The corollary is proved.  
\end{proof}


\end{document}